\documentclass[letterpaper, 10 pt, conference]{ieeetran}  

\IEEEoverridecommandlockouts                              

\newif\ifarxiv
\newif\ifnamechange
\newif\ifauthornames

\newif\ifimages
\imagestrue

\arxivtrue

\namechangefalse

\authornamestrue

\usepackage{makeidx}         
\usepackage{graphicx}        
\ifarxiv
\fi

\makeindex             

\usepackage[bookmarks=true]{hyperref}
\usepackage{mathtools}

\usepackage{enumerate}
\usepackage{epstopdf}
\usepackage{subfigure}
\usepackage{cleveref}

\usepackage{tablefootnote}

\usepackage{amsmath,amssymb}
\usepackage{algorithm}
\usepackage[noend]{algpseudocode}

\usepackage[usenames,dvipsnames,svgnames]{xcolor}

\newcommand{\V}{\mathcal V}
\newcommand{\E}{\mathcal E}
\renewcommand{\S}{\mathcal S}

\newcommand{\Vdo}{\mathcal V_{\text{or, def}}}
\newcommand{\Vdd}{\mathcal V_{\text{dest, def}}}

\graphicspath{{./fig/}}
\usepackage{color}
\usepackage{amsmath}
\usepackage{amssymb}
\usepackage{graphicx}
\usepackage{comment,xspace}
\usepackage{hyperref}
\usepackage{fancybox}

\newtheorem{theorem}{Theorem}[section]

\newtheorem{lemma}[theorem]{Lemma}

\newcommand{\real}{{\mathbb{R}}}
\newcommand{\reals}{\real}

\newcommand{\M}{\mathcal{M}}

\newcommand{\R}{\mathbb{R}}
\newcommand{\N}{\mathbb{N}}

\renewcommand{\baselinestretch}{0.951}

\title{\LARGE \bf
Routing Autonomous Vehicles in Congested Transportation Networks: Structural Properties and Coordination Algorithms
}

\ifnamechange
\title{\LARGE \bf
Autonomous Vehicle Routing in Congested Road Networks
}
\fi

\author{Authors Omitted for Blind Review
}

\ifnamechange
\author{Federico Rossi,\textsuperscript{*} Rick Zhang,\textsuperscript{*} and Marco Pavone%
\thanks{\textnormal{\textsuperscript{*}These authors contributed equally to this work.} Working paper, in preparation for journal submission. }
}
\fi

\ifauthornames
\author{Rick Zhang,\textsuperscript{*} Federico Rossi,\textsuperscript{*} and Marco Pavone%
\thanks{\textnormal{\textsuperscript{*}These authors contributed equally to this work.}} \thanks{This research was supported by National Science Foundation under CAREER Award CMMI-1454737 and by the Dr. Cleve B. Moler Stanford Graduate Fellowship. 
Rick Zhang, Federico Rossi, and Marco Pavone are with the Department of Aeronautics \& Astronautics, Stanford University, Stanford, CA 94305 {\tt\small \{rickz, frossi2, pavone\}@stanford.edu}}%
}
\fi

\begin{document}

\maketitle 
\thispagestyle{empty}
\pagestyle{empty}

\begin{abstract}
This paper considers the problem of routing and rebalancing a shared fleet of autonomous (i.e., self-driving) vehicles providing on-demand mobility within a {\em capacitated} transportation network, where congestion might disrupt throughput. We model the problem within a network flow framework and show that under relatively mild assumptions the rebalancing vehicles, if properly coordinated, do not lead to an increase in congestion (in stark contrast to common belief). From an algorithmic standpoint, such theoretical insight suggests that the problem of routing customers and rebalancing vehicles can be {\em decoupled}, which leads to a computationally-efficient routing and rebalancing algorithm for the autonomous vehicles. Numerical experiments and case studies corroborate our theoretical insights and show that the proposed algorithm  
outperforms state-of-the-art point-to-point methods by avoiding excess congestion on the road. Collectively, this paper provides a rigorous approach to the problem of congestion-aware, system-wide coordination of autonomously driving vehicles, and to the characterization of the sustainability of such robotic systems.

\end{abstract}

\section{Introduction}

Autonomous (i.e., robotic, self-driving) vehicles are rapidly becoming a reality and hold great promise for increasing safety and enhancing mobility for those unable or unwilling to drive \cite{WJM-CEBB-LDB:10, Google:14}. A particularly attractive operational paradigm involves coordinating a fleet of autonomous vehicles to provide on-demand service to customers, also called autonomous mobility-on-demand (AMoD). An AMoD system may reduce the cost of travel \cite{KS-KT-RZ-EF-DM-MP:14} as well as provide additional sustainability benefits such as increased overall vehicle utilization, reduced demand for urban parking infrastructure, and reduced pollution (with electric vehicles) \cite{WJM-CEBB-LDB:10}. The key benefits of AMoD are realized through vehicle sharing, where each vehicle, after servicing a customer, drives itself to the location of the next customer or \emph{rebalances} itself throughout the city in anticipation of future customer demand \cite{MP-SLS-EF-DR:12}. 

In terms of traffic congestion, however, there has been no consensus on whether autonomous vehicles in general, and AMoD systems in particular, will ultimately be beneficial or detrimental. It has been argued that by having faster reaction times, autonomous vehicles may be able to drive faster and follow other vehicles at closer distances without compromising safety, thereby effectively increasing the capacity of a road and reducing congestion. They may also be able to interact with traffic lights to reduce full stops at intersections \cite{JP-FS-VM-AJ-JCD-TDP:10}. On the downside, the process of vehicle rebalancing (empty vehicle trips) increases the total number of vehicles on the road (assuming the number of vehicles with customers stays the same). Indeed, it has been argued that the presence of many rebalancing vehicles may contribute to an \emph{increase}  in congestion \cite{BT:15,MB:16}. These statements, however, do not take into account that in an AMoD system the operator has control over the actions (destination and routes) of the vehicles, and may route vehicles intelligently to avoid increasing congestion or perhaps even decrease it. 

Accordingly, the goal of this paper is twofold. First, on an engineering level, we aim to devise routing and rebalancing algorithms for an autonomous vehicle fleet that seek to minimize congestion. Second, on a socio-economic level, we aim to rigorously address the concern that autonomous cars may lead to increased congestion and thus disrupt currently congested transportation infrastructures.

\emph{Literature review:} In this paper, we investigate the problem of controlling an AMoD system within a road network in the presence of congestion effects. Previous work on AMoD systems have primarily concentrated on the rebalancing problem \cite{MP-SLS-EF-DR:12, KS-KT-RZ-EF-DM-MP:14}, whereby one strives to allocate empty vehicles throughout a city while minimizing fuel costs or customer wait times. The rebalancing problem has been studied in \cite{MP-SLS-EF-DR:12} using a fluidic model and in \cite{RZ-MP:15} using a queueing network model. An alternative formulation is the one-to-one pickup and delivery problem \cite{GB-JFC-GL:10}, where a fleet of vehicles service pickup and delivery requests within a given region. Combinatorial asymptotically optimal algorithms for pickup and delivery problems were presented in \cite{KT-MP-EF:11, KT-MP-EF:13}, and generalized to road networks in \cite{KT-MP-EF:12a}. Almost all current approaches assume point-to-point travel between origins and destinations (no road network), and even routing problems on road networks (e.g. \cite{KT-MP-EF:12a}) do not take into account vehicle-to-vehicle interactions that would cause congestion and reduce system throughput. 

On the other hand, traffic congestion has been studied in economics and transportation for nearly a century. The first congestion models \cite{JGW:52, MJL-GBW:55, CFD:94} sought to formalize the relationship between vehicle speed, density, and flow. Since then, approaches to modeling congestion have included empirical \cite{BSK:09}, simulation-based \cite{MT-AH-DH:00,QY-HNK:96,MB-MR-KM-DC-NL-KN-KA:09}, queueing-theoretical \cite{CO-MB:09}, and optimization \cite{SP-HSM:95,BNJ:91}. While there have been many high fidelity congestion models that can accurately predict traffic patterns, the primary goal of congestion modeling has been the \emph{analysis} of traffic behavior. Efforts to \emph{control} traffic have been limited to the control of intersections \cite{TL-PK-NW-HLV-LLHA-SSPH:15,NX-EF-YL-YL-YW-DW:15} and freeway on-ramps \cite{MP-HHS-JMB:91} because human drivers behave non-cooperatively.
The problem of cooperative, system-wide routing (a key benefit of AMoD systems) is similar to the dynamic traffic assignment problem (DTA) \cite{BNJ:91} and to \cite{DW-JPVDB-MCL-DM:11,DW-CB-MCL:14} in the case of online routing. 
The key difference is that these approaches only optimize routes for passenger vehicles while we seek to optimize the routes of {\em both} passenger vehicles \emph{and} empty rebalancing vehicles. 

\emph{Statement of contributions:} The contribution of this paper is threefold. First, we model an AMoD system within a network flow framework, whereby customer-carrying and empty rebalancing vehicles are represented as flows over a  {\em capacitated} road network (in such model, when the flow of vehicles along a road reaches a critical capacity value, congestion effects occur). Within this model, we provide a cut condition for the road graph that needs to be satisfied for congestion-free customer and rebalancing flows to exist. Most importantly, under the assumption of a {\em symmetric} road network, we investigate an existential result that leads to two key conclusions: (1) rebalancing does not increase congestion, and (2) for certain cost functions, the problems of finding customer and rebalancing flows can be decoupled. Second, leveraging the theoretical insights, we propose a computationally-efficient algorithm for congestion-aware routing and rebalancing of an AMoD system that is broadly applicable to time-varying, possibly asymmetric road networks. Third, through numerical studies on real-world traffic data, we validate our assumptions and show that the proposed real-time routing and rebalancing algorithm outperforms state-of-the-art point-to-point rebalancing algorithms in terms of lower customer wait times by avoiding excess congestion on the road.

\emph{Organization:} The remainder of this paper is organized as follows: in Section \ref{sec:model} we present a network flow model of an AMoD system on a capacitated road network and formulate the routing and rebalancing problem. In Section \ref{sec:analysis} we present key structural properties of the model including fundamental limitations of performance and conditions for the existence of feasible (in particular, congestion-free) solutions. 
The insights from Section \ref{sec:analysis} are used to develop a practical real-time routing and rebalancing algorithm in Section \ref{sec:realtime}. Numerical studies and simulation results are presented in Section \ref{sec:num}, and in Section \ref{sec:conc} we draw conclusions and discuss directions for future work. 

\section{Model Description and Problem Formulation} \label{sec:model}

In this section we formulate a network flow model for an AMoD system operating over a capacitated road network. The model allows us to derive key structural insights into the vehicle routing and rebalancing problem, and motivates the design of real-time, congestion-aware algorithms for coordinating the robotic vehicles. We start in Section \ref{subsec:cong_model} with a discussion of our congestion model; then, in Section \ref{sec:fluidmodel} we provide a detailed description of the overall AMoD system model.

\subsection{Congestion Model}\label{subsec:cong_model}
We use a simplified congestion model consistent with classical traffic flow theory \cite{JGW:52}. In classical traffic flow theory, at low vehicle densities on a road link, vehicles travel at the free flow speed of the road (imposed by the speed limit). This is referred to as the free flow phase of traffic. In this phase, the free flow speed is approximately constant \cite{BSK:09a}. The flow, or flow rate, is the number of vehicles passing through the link per unit time, and is given by the product of the speed and density of vehicles. When the flow of vehicles reaches an empirically observed critical value, the flow reaches its maximum. Beyond the critical flow rate, vehicle speeds are dramatically reduced and the flow decreases, signaling the beginning of traffic congestion. The maximum stationary flow rate is called the \emph{capacity} of the road link in the literature. In our approach, road capacities are modeled as {\em constraints on the flow of vehicles}. In this way, the model captures the behavior of vehicles up to the onset of congestion. 

This simplified congestion model is adequate for our purposes because the goal is not to analyze the behavior of vehicles in congested networks, but to control vehicles in order to avoid the onset of congestion. We also do not explicitly model delays at intersections, spillback behavior due to congestion, or bottleneck behavior due to the reduction of the number of lanes on a road link. An extension to our model that accommodates (limited) congestion on links is presented in Section \ref{sec:congrelax}. 

\subsection{Network Flow Model of AMoD system}
\label{sec:fluidmodel}

We consider a road network modeled as a directed graph $G = (\V,\E)$, where $\V$ denotes the node set and $\E\subseteq \V\times \V$ denotes the edge set. Figure \ref{fig:roadnet} shows one such network. The nodes $v$ in $\V$ represent intersections and locations for trip origins/destinations, and the edges $(u,v)$ in $\E$ represent road links. As discussed in Section \ref{subsec:cong_model}, congestion is modeled by imposing capacity constraints on the road links: each constraint represents the capacity of the road upon the onset of congestion. Specifically, for each road link $(u,v) \in \E$, we denote by $c(u,v): \E \mapsto \mathbb{N}_{> 0}$ the capacity of that link. 
When the flow rate on a road link is less than the capacity of the link, all vehicles are assumed to travel at the free flow speed, or the speed limit of the link. For each road link $(u,v) \in \E$, we denote by $t(u,v): \E \mapsto \mathbb{R}_{\geq  0}$ the corresponding free flow time required to traverse road link $(u,v)$. Conversely, when the flow rate on a road link is larger than the capacity of the link, the traversal time  is assumed equal to $\infty$ (we reiterate that our focus in this section is on avoiding the onset of congestion).

We assume that the road network is {\em capacity-symmetric} (or symmetric for short): for any cut\footnote{For any subset of nodes $\S\subseteq \V$, we define a \emph{cut} $(\S, \bar \S)\subseteq \E$ as the set of edges whose origin lies in $\S$ and whose destination lies in $\bar \S=\{\V\setminus \S\}$. Formally, $(\S,\bar \S):=\{(u,v)\in \E: u\in \S, v \in \bar \S\}$.} $(\S, \bar \S)$ of $G(\V, \E)$, the overall capacity of the edges connecting nodes in $\S$ to nodes in $\bar \S$ equals the overall capacity of the edges connecting nodes in $\bar \S$ to nodes in $ \S$, that is
\[
 \sum_{(u,v)\in \E: \,\,u\in \S, \, v\in \bar S} c(u,v) = \sum_{(v,u)\in \E:\,\,u\in \S, \, v\in \bar S} c(v,u) 
\]
It is easy to verify that a network is capacity-symmetric if and only if the overall capacity entering each \emph{node} equals the capacity exiting each node., i.e.
\[
\sum_{u\in \V : (u,v)\in \E} c(u,v) = \sum_{w\in \V: (v,w)\in \E} c(v,w)
\]
If all \emph{edges} have symmetrical capacity, i.e., for all $ (u,v)\in \E$, $c(u,v)=c(v,u)$, then the network is capacity-symmetric. The converse statement, however,  is not true in general.

Transportation requests are described by the tuple $(s, t, \lambda)$, where $s \in \V$ is the origin of the requests, $t\in \V$ is the destination, and $\lambda \in \R_{>0}$ is the rate of requests, in customers per unit time. Transportation requests are assumed to be stationary and deterministic, i.e., the rate of requests does not change with time and is a deterministic quantity. The set of transportation requests is denoted by $\mathcal M = \{(s_m, t_m, \lambda_m) \}_m$, and its cardinality is denoted by $M$. 

Single-occupancy vehicles travel within the network while servicing the transportation requests. We denote $f_m(u,v): \E \mapsto \R_{\geq 0}$, $m=\{1,\ldots, M\}$, as the \emph{customer flow} for requests $m$ on edge $(u,v)$, i.e., the amount of flow from origin $s_m$ to destination $t_m$ that uses link $(u,v)$. We also denote $f_{R}(u,v): \E\mapsto \R_{\geq 0}$ as the \emph{rebalancing flow} on edge $(u,v)$, i.e., the amount of rebalancing flow traversing edge $(u,v)$ needed to realign the vehicles with the asymmetric distribution of transportation requests. 

\subsection{The Routing Problem}
\label{sec:routingproblem}
The goal  is to compute flows for the autonomous vehicles that (i) transfer customers to their desired destinations in minimum time (customer-carrying trips) and (ii) rebalance vehicles throughout the network to realign the vehicle fleet with transportation demand (customer-empty trips). Specifically, the  {\em Congestion-free Routing and Rebalancing Problem (CRRP)} is formally defined as follows. Given a capacitated, symmetric network $G(\V,\E)$, a set of transportation requests $\mathcal M = \{(s_m, t_m, \lambda_m)\}_m$, and a weight factor $\rho>0$, solve 

{\small
\begin{align}
&\!\!\!\!\underset{f_m(\cdot, \cdot), f_R(\cdot,\cdot)}{\text{minimize}} && \!\!\!\!\!\!\!\sum_{m\in\mathcal{M}}  \sum_{(u,v)\in \E} t(u,v) f_m(u,v) \!+\! \rho \sum_{(u,v) \in \E} t(u,v) f_R(u,v) \label{eqn:pdcost} \\
&\text{subject to} && \!\!\!\!\!\!\!\!\sum_{u\in \V} f_m(u,s_m) + \lambda_m = \sum_{w\in \V} f_m(s_m,w) \;\; \forall m\in \mathcal{M} 
\label{eqn:pdmbal}\\
& && \!\!\!\!\!\!\!\!\sum_{u\in \V} f_m(u,t_m) = \lambda_m + \sum_{w\in \V} f_m(t_m,w) \;\; \forall m\in \mathcal{M} \label{eqn:pdmsource}\\
& && \!\!\!\!\!\!\!\!\sum_{u\in \V} f_m(u,v) =  \sum_{w\in \V} f_m(v,w)  \nonumber \\ 
& && \!\!\!\!\!\!\!\!\qquad \qquad \qquad \qquad \forall m\in \mathcal{M}, v\in \V\setminus\{ s_m, t_m\} \label{eqn:pdmsink} \\
& && \!\!\!\!\!\!\!\!\sum_{u\in \V} f_R(u,v) +\sum_{m\in \mathcal M} 1_{v = t_m}\lambda_m \nonumber \\
& && \!\!\!\!\!\!\!\!=  \sum_{w\in \V} f_R(v,w) + \sum_{m\in \mathcal M} 1_{v = s_m}\lambda_m \quad \forall v \in \V \label{eqn:pdrbal}\\
& && \!\!\!\!\!\!\!\!f_R(u,v) + \sum_{m\in \mathcal M} f_m(u,v) \leq c(u,v)  \;\; \forall (u,v)\in \E \label{eqn:pdcong}
\end{align}
}

\begin{figure}[t]
\centering
\includegraphics[width=.22\textwidth]{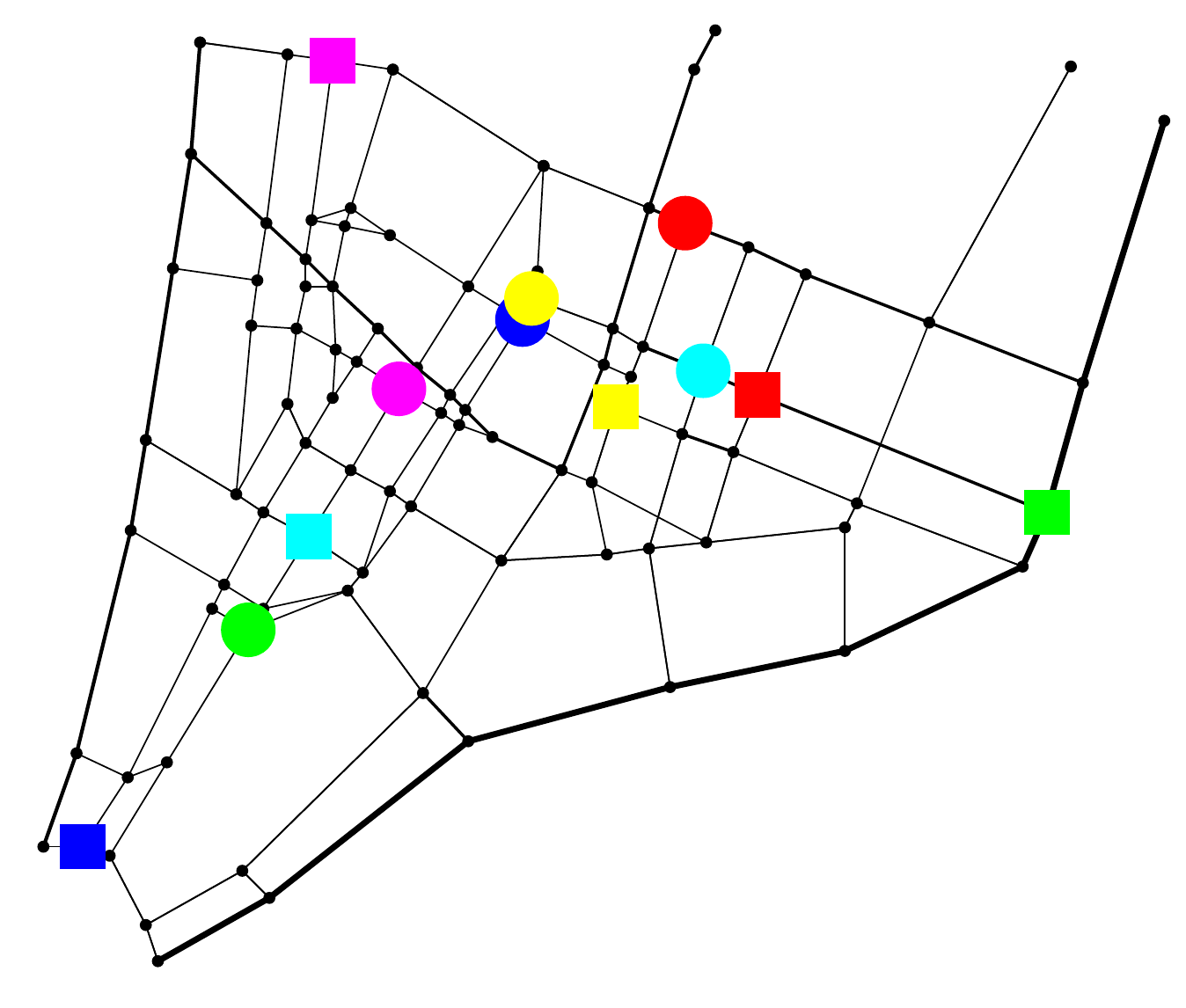}
\vspace{-1em}
\caption{A road network modeling Lower Manhattan and the Financial District. Nodes (denoted by small black dots) model intersections; select nodes, denoted by colored circular and square markers, model passenger trips' origins and destinations. Different trip requests are denoted by different colors. Roads are modeled as edges; line thickness is proportional to road capacity. }
\label{fig:roadnet}
\end{figure}

The cost function \eqref{eqn:pdcost} is a weighted sum (with weight $\rho$) of the overall duration of all passenger trips and the duration of rebalancing trips. Constraints \eqref{eqn:pdmbal}, \eqref{eqn:pdmsource} and \eqref{eqn:pdmsink} enforce continuity of each trip (i.e., flow conservation) across nodes. Constraint \eqref{eqn:pdrbal} ensures that vehicles are rebalanced throughout the road network to re-align vehicle distribution with transportation requests, i.e. to ensure that every outbound customer flow is matched by an inbound flow of rebalancing vehicles and vice versa.  Finally, constraint \eqref{eqn:pdcong} enforces the capacity constraint on each link (function $1_{x}$ denotes the indicator function of the Boolean variable $x = \{\text{true, false}\}$, that is $1_{x}$ equals one if $x$ is true, and equals zero if $x$ is false).
Note that the CRRP is a linear program and, in particular, a special instance of the fractional multi-commodity flow problem \cite{RKA-TLM-JBO:93}.  

We denote a customer flow $\{f_m(u,v)\}_{(u,v),m}$ that satisfies Equations \eqref{eqn:pdmbal}, \eqref{eqn:pdmsource}, \eqref{eqn:pdmsink} and \eqref{eqn:pdcong} as a \emph{feasible customer flow}. For a given set of feasible customer flows $\{f_m(u,v)\}_{(u,v),m}$, we denote a flow $\{f_R(u,v)\}_{(u,v)}$ that satisfies Equation \eqref{eqn:pdrbal} and such that the combined flows $\{f_m(u,v), f_R(u,v)\}_{(u,v),m}$ satisfy Equation \eqref{eqn:pdcong} as a \emph{feasible rebalancing flow}. We remark that a rebalancing flow that is feasible with respect to a set of customer flows may be infeasible for a different collection of customer flows.

For a given set of optimal flows $\{f_m^*(u,v)\}_{(u,v),m}$ and $\{f_R^*(u,v)\}_{(u,v)}$, the minimum number of vehicles needed to implement them is given by

\vspace{-.5em}
{\small
\begin{equation*}
V_{\text{min}} =\left \lceil \sum_{m \in \mathcal{M}} \sum_{(u,v) \in \E} t(u,v)\Big(f_m^*(u,v) + f_R^*(u,v)\Big) \right \rceil.
\end{equation*}
}
This follows from a similar analysis done in \cite{MP-SLS-EF-DR:12} for point-to-point networks. Hence, the cost function \eqref{eqn:pdcost} is aligned with the desire of minimizing the number of vehicles needed to operate an AMoD system.

\subsection{Discussion}

A few comments are in order. First, we assume that transportation requests are time invariant. This assumption is valid when transportation requests change slowly with respect to the average duration of a customer's trip, which is often the case in dense urban environments \cite{HN:71}.  Additionally, in Section \ref{sec:realtime} we will present algorithmic tools that allow one to extend the insights gained from the time-invariant case to the time-varying counterpart. Second, the assumption of single-occupancy for the vehicles models most of the existing (human) one-way vehicle sharing systems (where the driver is considered ``part" of the vehicle), and chiefly disallows the provision of  ride-sharing or carpooling service (this is an aspect left for future research).
Third, as also discussed in Section \ref{subsec:cong_model}, our congestion model is simpler and less accurate than typical congestion models used in the transportation community. However, our model lends itself to efficient real-time optimization and thus it is well-suited to the \emph{control} of fleets of autonomous vehicles. Existing high-fidelity congestion models should be regarded as complementary and could be used offline to identify the congestion thresholds used in our model.
Fourth, while we have defined the CRRP in terms of fractional flows, an integer-valued counterpart can be defined and (approximately) solved to find optimal routes for each {\em individual} customer and vehicle. Algorithmic aspects will be investigated in depth in Section \ref{sec:realtime}, with the goal of devising practical, real-time routing and rebalancing algorithms.  Fifth, trip requests are assumed to be known. In practice, trip requests can be reserved in advance, estimated from historical data, or estimated in real time.  Finally, the assumption of capacity-symmetric road networks indeed appears reasonable for a number of major U.S. metropolitan areas (note that this assumption is much less restrictive than assuming every \emph{individual} road is capacity-symmetric). In the 
\ifarxiv
Supplementary Material,
\else
extended version of this paper \cite{RZ-FR-MP:16aEV},
\fi
by using OpenStreetMap data \cite{MH-PW:08}, we provide a rigorous characterization in terms of capacity symmetry of the road networks of New York City, Chicago, Los Angeles and other major U.S. cities. The results consistently show that urban road networks are usually symmetric to a {\em very high} degree.
 Additionally, several of  our theoretical and algorithmic results extend to the case where this assumption is lifted, as it will be highlighted throughout the paper.

\section{Structural Properties of the \\Network Flow Model} \label{sec:analysis}
In this section we provide two key structural results for the network flow model presented  in Section \ref{sec:fluidmodel}. First, we provide a cut condition that needs to be satisfied for feasible customer and rebalancing flows to exist. In other words, this condition provides a fundamental limitation of performance for congestion-free AMoD service in a given road network. Second, we investigate an existential result (our main theoretical result) that is germane to two key conclusions: (1) rebalancing does not increase congestion in symmetric road networks, and (2) for certain cost functions, the problems of finding  customer and rebalancing flows can be {\em decoupled} -- an insight that will be heavily exploited in subsequent sections.

\subsection{Fundamental Limitations}
 
We start with a few definitions. For a given set of feasible customer flows $\{f_m(u,v)\}_{(u,v),m}$, we denote by $F_{\text{out}}(\S, \bar{\S})$ the overall flow exiting a cut $(\S, \bar{\S})$, i.e., $F_{\text{out}}(\S, \bar{\S}) := \sum_{m\in \mathcal{M}}\sum_{u \in \S, v \in \bar{\S}} f_m(u,v)$. Similarly, we denote by $C_{\text{out}}(\S, \bar{\S})$ the capacity of the network exiting $\S$, i.e., $C_{\text{out}}(\S, \bar{\S}) = \sum_{u \in S, v \in \bar{\S}} c(u,v)$. Analogously, $F_{\text{in}}(\S, \bar{\S})$ denotes  the overall flow entering $\S$ from $\bar{\S}$, i.e., $F_{\text{in}}(\S, \bar{\S}):=F_{\text{out}}(\bar{\S}, \S)$, and $C_{\text{in}}(\S, \bar{\S})$ denotes the capacity entering $\S$ from $\bar{\S}$, i.e., $C_{\text{in}}(\S, \bar{\S}):=C_{\text{out}}(\bar{\S}, {\S})$. We highlight that the arguments leading to the main result of this subsection (Theorem \ref{thm:feasreb}) do not require the assumption of capacity symmetry; hence, Theorem \ref{thm:feasreb} holds for {\em asymmetric} road networks as well. 
 
The next technical lemma (whose proof is provided in the 
\ifarxiv
 Supplementary Material)
\else
 extended version of this paper \cite{RZ-FR-MP:16aEV})
 \fi
 shows that the net flow leaving set $\S$ equals the difference between the flow originating from the origins $s_m$ in $\S$ and the flow exiting through the destinations $t_m$ in $\S$, that is,
 \begin{lemma}[Net flow across a cut]\label{lemma:netflow}
Consider a set of feasible customer flows $\{f_m(u,v)\}_{(u,v),m}$. Then, for every cut $(\S, \bar \S)$, the net flow leaving set $\S$ satisfies
 \[
 F_{\text{out}}(\S, \bar{\S}) - F_{\text{in}}(\S, \bar{\S})= \sum_{m\in \M} 1_{s_m\in \S}\lambda_m - \sum_{m\in \M} 1_{t_m\in\S}\lambda_m.
 \]
 \end{lemma}

We now state two additional lemmas (whose proofs are given in 
\ifarxiv
the Supplementary Material)
\else
\cite{RZ-FR-MP:16aEV})
\fi
 providing, respectively, lower and upper bounds for the outflows $ F_{\text{out}}(\S, \bar{\S})$.
 
\begin{lemma}[Lower bound for outflow] \label{lemma:grossflow}
Consider a set of feasible customer flows $\{f_m(u,v)\}_{(u,v),m}$. Then, for any cut $(\S, \bar \S)$, the overall flow $F_{\text{out}}(\S, \bar{\S})$ exiting cut $(\S, \bar \S)$ is lower bounded according to
 \[
\sum_{m \in \M} 1_{s_m\in \S, t_m \in \bar \S} \lambda_m \leq F_{\text{out}}(\S, \bar{\S}).
 \]
 \end{lemma}

\begin{lemma}[Upper bound for outflow] \label{lemma:necessary}
Assume there exists a set of {\em feasible} customer and rebalancing flows $\{f_m(u,v), \, f_R(u,v)\}_{(u,v),m}$. Then, 
for every cut $(\S,\bar{\S})$, 
\begin{enumerate}
\item $F_{\text{out}}(\S, \bar{\S}) \leq C_{\text{out}}(\S, \bar{\S})$, and
\item $F_{\text{out}}(\S, \bar{\S}) \leq C_{\text{in}}(\S, \bar{\S})$.
\end{enumerate}
\end{lemma}

We are now in a position to present a {\em structural} (i.e., flow-independent) necessary condition for the existence of feasible customer and rebalancing flows.  

\begin{theorem}[Necessary condition for feasible flows] \label{thm:feasreb}
A necessary condition for the existence of a set of {\em feasible} customer and rebalancing flows $\{f_m(u,v), \, f_R(u,v)\}_{(u,v),m}$, is that, for every cut $(\S, \bar{\S})$,
\begin{enumerate}
\item $\sum_{m \in \M} 1_{s_m\in \S, t_m \in \bar \S} \lambda_m \leq C_{\text{out}}(\S, \bar{\S})$, and
\item $\sum_{m \in \M} 1_{s_m\in \S, t_m \in \bar \S} \lambda_m \leq C_{\text{in}}(\S, \bar{\S})$.
\end{enumerate}
\end{theorem}
\begin{proof}
The theorem is a trivial consequence of Lemmas \ref{lemma:grossflow}  and \ref{lemma:necessary}.
\end{proof}

Theorem \ref{thm:feasreb} essentially provides a structural fundamental limitation of performance for a given road network: if the cut conditions in Theorem \ref{thm:feasreb} are not met, then there is no hope of finding congestion-free customer and rebalancing flows. We reiterate that Theorem \ref{thm:feasreb} holds for both symmetric and asymmetric networks (for a symmetric network, claim 2) in Lemma \ref{lemma:necessary} and condition 2) in  Theorem \ref{thm:feasreb} are redundant).

\subsection{Existence of Congestion-Free Flows}
In this section we address the following question: assuming there exists a feasible customer flow, is it always possible to find a feasible rebalancing flow? As we will see, the answer to this question is affirmative and has both conceptual and algorithmic implications. 

\begin{theorem}[Feasible rebalancing]
\label{thm:symmetric}
Assume there exists a set of feasible customer  flows $\{f_m(u,v)\}_{(u,v),m}$. Then, it is {\em always} possible to find a set of feasible rebalancing flows $\{f_R(u,v)\}_{(u,v)}$.

\end{theorem}
\begin{proof}
We prove the theorem for the special case where no node $v\in \V$ is associated with both an origin and a destination for the transportation requests in $\mathcal M$. This is without loss of generality, as the general case where a node $v$ has  both an origin and a destination assigned can be reduced to this special case, by associating with node $v$ a ``shadow" node so that (i) all destinations are assigned to the shadow node and (ii) node $v$ and its shadow node are mutually connected via an infinite-capacity, zero-travel-time edge.

We start the proof by defining the concepts of \emph{partial rebalancing flows} and {\em defective origins and destinations}. Specifically, a partial rebalancing flow, denoted as $\{\hat f_R(u,v)\}_{(u,v)}$, is a set of mappings from $\E$ to $\reals_{\geq 0}$ obeying  the following properties:
\begin{enumerate}
\item It satisfies constraint \eqref{eqn:pdrbal} at every node that is not an origin nor a destination, that is $\forall \, v\in \{\V\setminus\{ \{s_m\}_m \cup \{t_m\}_m\}\}$,
\begin{align*}
 \sum_{u\in \V} \hat f_R(u,v) =  \sum_{w\in \V} \hat f_R(v,w). 
 \end{align*}
 \item It violates constraint \eqref{eqn:pdrbal} in the ``$\leq$ direction" at every node that is an origin, that is $\forall \, v\in \V \text{ such that } \, \exists m\in \M: v=s_m$,
\begin{align*}
 \sum_{u\in \V} \hat f_R(u,v)
 \leq  \sum_{w\in \V} \hat f_R(v,w) + \sum_{m\in \mathcal M} 1_{v = s_m}\lambda_m.
\end{align*} 

 \item  It violates constraint \eqref{eqn:pdrbal} in the ``$\geq$ direction" at every node that is a destination, that is $\forall \,v\in \V \text{ such that } \, \exists m\in \M: v=t_m$,
\begin{align*} \sum_{u\in \V} \hat f_R(u,v) +\sum_{m\in \mathcal M} 1_{v = t_m}\lambda_m
 \geq  \sum_{w\in \V} \hat f_R(v,w).
\end{align*} 
\item The combined customer and partial rebalancing flows $\{f_m(u,v),\hat f_R(u,v)\}_{(u,v),m}$ satisfy Equation \eqref{eqn:pdcong} for every edge $(u,v)\in \E$.
\end{enumerate}
Note that the trivial zero flow, that is  $\hat f_R(u,v)=0$ for all $(u,v)\in \E$, is a partial rebalancing flow (in other words, the set of partial rebalancing flows in not empty). Clearly a feasible rebalancing flow is also a partial rebalancing flow, but the opposite is not necessarily true.

For a given partial rebalancing flow, we denote an origin node, that is a node $v\in \V$ such that $v = s_m$ for some $m=1,\ldots, M$, as a \emph{defective} origin if Equation \eqref{eqn:pdrbal} is not satisfied at $v=s_m$ (in other words, the strict inequality $<$ holds). Analogously, we denote a destination node, that is a node $v\in \V$ such that $v = t_m$ for some $m=1,\ldots, M$, as a \emph{defective} destination if Equation \eqref{eqn:pdrbal} is not satisfied at $v=t_m$ (in other words, the strict inequality $>$ holds). The next lemma (whose proof is provided in 
\ifarxiv
the Supplementary Material)
\else
\cite{RZ-FR-MP:16aEV})
\fi
links the concepts of partial rebalancing flows and defective origins/destinations.

\begin{lemma}[Co-existence of defective origins/destinations]
\label{lemma:prsourcesink}
For every partial rebalancing flow that is not a feasible rebalancing flow, there exists at least one node $u\in \V$ that is a defective origin, {\em and} one node $v\in \V$ that is a defective destination.
\end{lemma}

For a given set of customer flows $\{f_m(u,v)\}_{(u,v),m}$ and partial rebalancing flows $\{\hat f_R(u,v)\}_{(u,v)}$, we call an edge $(u,v)\in \E$ \emph{saturated} if Equation  \eqref{eqn:pdcong} holds with equality for that edge. We call a path \emph{saturated} if at least one of the edges along the  path is saturated. We now prove the existence of a special partial rebalancing flow where defective destinations and defective origins are separated by a graph cut formed exclusively by saturated edges (this result, and its consequences, are illustrated in Figure \ref{fig:UnbalancedCut}). 

\begin{lemma}
\label{lemma:satcut}
Assume there exists a set of {\em feasible} customer  flows $\{f_m(u,v)\}_{(u,v),m}$, but there does not exist a set of feasible rebalancing flows $\{f_R(u,v)\}_{(u,v)}$. Then, there exists a partial rebalancing flow $\{\hat f_R(u,v)\}_{(u,v)}$ that induces a graph cut $(\S, \bar \S)$ with the following properties: (i) all defective destinations are in $\S$, (ii) all defective origins are in  $\bar \S$, and (iii) all edges in $(S, \bar \S)$ are saturated.
\end{lemma}

\begin{figure}[t]
\centering
\includegraphics[width=0.22\textwidth]{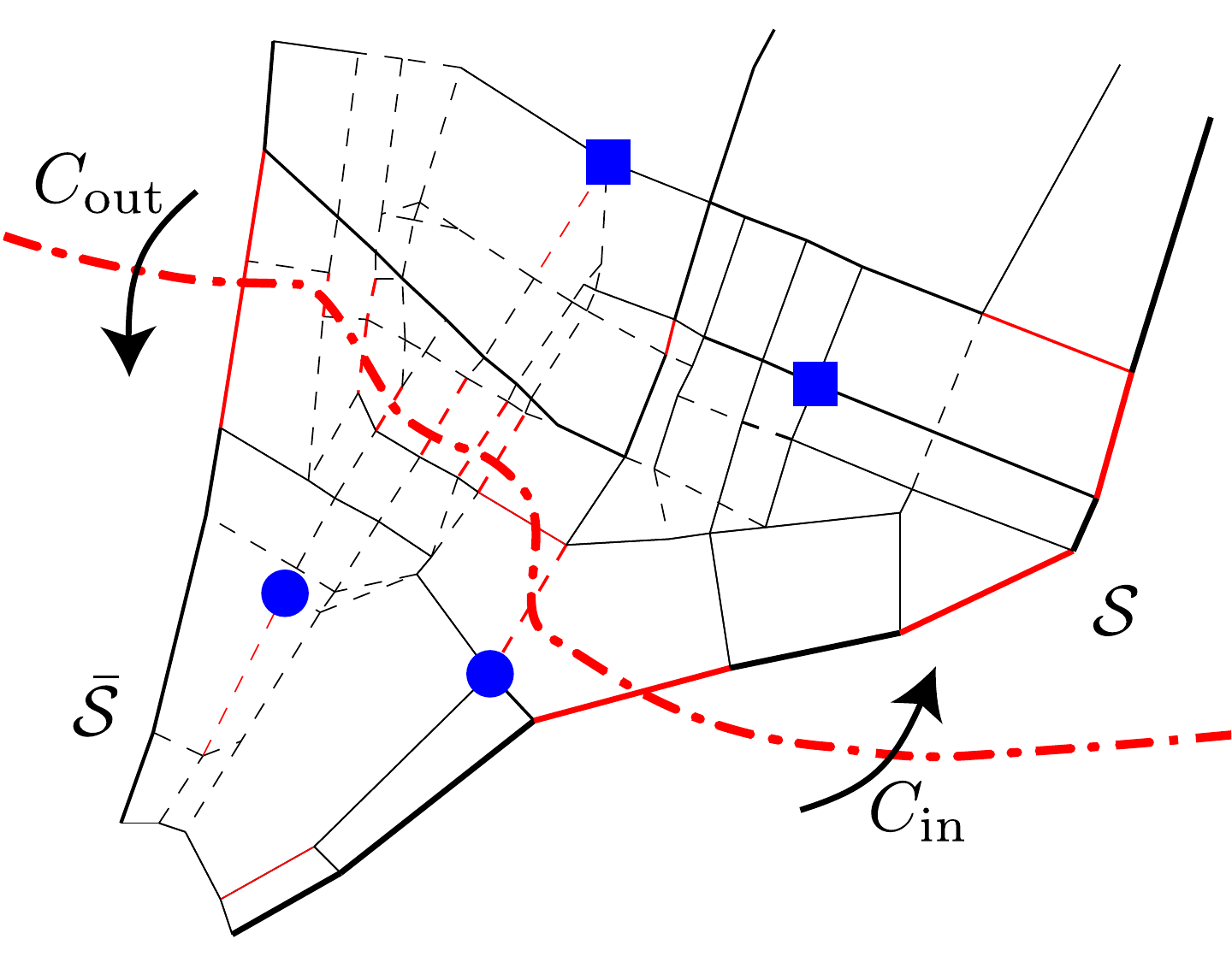}
\vspace{-1em}
\caption{A graphical representation of Lemma \ref{lemma:satcut}. If there exists a set of feasible customer flows but there does not exist a set of
feasible rebalancing flows, one can find a partial rebalancing flow where all the defective origins, represented as blue circles, are separated from all the defective destinations, represented as blue squares, by a cut of saturated edges (shown in red). Note that not all saturated edges necessarily  belong to the cut. In the proof of Theorem \ref{thm:symmetric} we show that the capacity of such a cut $(\S,\bar \S)$ is asymmetric, i.e., $C_\text{out}<C_\text{in}$ -- a contradiction that leads to the claim of Theorem \ref{thm:symmetric}.}

\label{fig:UnbalancedCut}
\end{figure}

We are now in a position to prove Theorem \ref{thm:symmetric}. The proof is by contradiction. Assume that a set of feasible rebalancing flows $\{f_R(u,v)\}_{(u,v)}$ does not exist.  Then Lemma \ref{lemma:satcut} shows that there exists a partial rebalancing flow $\{\hat f_R(u,v)\}_{(u,v)}$ and a cut ($\S,\bar \S$) such that all defective destinations under  $\{\hat f_R(u,v)\}_{(u,v)}$ belong to $\S$ and all defective origins belong to $\bar \S$. Let us denote the sum of all partial rebalancing flows across cut $(\S, \bar \S)$ as
\[
\hat F^{\text{reb}}_{\text{out}}(\S,\bar{\S}):=\sum_{u\in \S,v\in\bar\S} \hat f_R(u,v),
\]
and, analogously, define $\hat F^{\text{reb}}_\text{in}(\S,\bar \S):=\hat F^{\text{reb}}_\text{out}(\bar \S, \S)$.
Since all edges in the cut $(\S, \bar \S)$ are saturated under  $\{\hat f_R(u,v)\}_{(u,v)}$, one has, due to Equation  \eqref{eqn:pdcong}, the equality
\[
C_\text{out}(\S,\bar{\S})=F_{\text{out}}(\S,\bar{\S})+ \hat F^{\text{reb}}_{\text{out}}(\S,\bar{\S}).
\]
Additionally, again due to Equation \eqref{eqn:pdcong}, one has the inequality
\[
F_{\text{in}}(\S,\bar{\S})+ \hat F^{\text{reb}}_{\text{in}}(\S,\bar{\S})\leq C_{\text{in}}(\S,\bar{\S}).
\]
Combining the above equations, one obtains

\vspace{-0.7em}
{\small
\[
F_{\text{in}}(S,\bar{S})+ \hat F^{\text{reb}}_{\text{in}}(\S,\bar{\S}) -F_{\text{out}}(\S,\bar{\S})- \hat F^{\text{reb}}_{\text{out}}(\S,\bar{\S}) \leq C_{\text{in}}(\S,\bar{\S}) - C_{\text{out}}(\S,\bar{\S}).
\]
}
To compute $\hat F^{\text{reb}}_{\text{in}}(\S,\bar{\S}) -\hat F^{\text{reb}}_{\text{out}}(\S,\bar{\S})$, we follow a procedure similar to the one used in Lemma \ref{lemma:netflow}. Summing  Equation \eqref{eqn:pdrbal} over all nodes in $\S$, one obtains,

\vspace{-0.6em}
{\small
\begin{align*}
&\sum_{v\in \S} \left[ \sum_{u\in \V} \hat f_R(u,v) + \sum_{m\in \M} 1_{v=t_m}\lambda_m \right]\\
&\qquad\qquad >\sum_{v\in \S} \left[\sum_{w\in\V}\hat f_R(v,w) + \sum_{m\in \M} 1_{v=s_m} \lambda_m \right]. 
\end{align*}
}
The strict inequality is due to the fact that for a partial rebalancing flow that is not feasible there exists at least one defective destination (Lemma \ref{lemma:prsourcesink}), which, by construction, must belong to $\S$. Simplifying those flows $\hat f_R(u,v)$ for which both $u$ and $v$ are in $\S$ (as such flows appear on both sides of the above inequality), one obtains
\begin{align*}
\hat F^{\text{reb}}_{\text{in}}(\S,\bar{\S}) -\hat F^{\text{reb}}_{\text{out}}(\S,\bar{\S}) >& \sum_{m\in \M}  1_{s_m\in \S}\lambda_m- \sum_{m\in \M } 1_{t_m \in \S}\lambda_m. 
\end{align*}
Also, by Lemma \ref{lemma:netflow},  
\[
F_{\text{out}}(\S, \bar{\S}) - F_{\text{in}}(\S, \bar{\S})= \sum_{m\in \M }1_{s_m \in \S} \lambda_m - \sum_{m\in \M }1_{t_m \in \S} \lambda_m.
\]
Collecting all the results so far, we conclude that 
\begin{align*}
0&<F_{\text{in}}(S,\bar{S})+ \hat F^{\text{reb}}_{\text{in}}(\S,\bar{\S}) -F_{\text{out}}(\S,\bar{\S})- \hat F^{\text{reb}}_{\text{out}}(\S,\bar{\S})\\
&= C_{\text{in}}(\S,\bar{\S}) - C_{\text{out}}(\S,\bar{\S}).
\end{align*}

Hence, we reached the conclusion that $C_{\text{in}}(\S,\bar{\S}) - C_{\text{out}}(\S,\bar{\S})>0$, or, in other words, the capacity of graph $G(\V, \E)$ across cut $(\S, \bar \S)$ is {\em not} symmetric. This contradicts the assumption that graph $G(\V, \E)$ is capacity-symmetric, and the claim follows.
\end{proof}

The importance of Theorem \ref{thm:symmetric} is twofold. First, perhaps surprisingly, it shows that for symmetric road networks it is {\em always} possible to rebalance the autonomous vehicles {\em without} increasing congestion -- in other words, the rebalancing of autonomous vehicles in a symmetric road network does {\em not} lead to an increase in congestion.
Second, from an algorithmic standpoint, if the cost function in the CRRP only depends on the customer flows (that is, $\rho = 0$ and the goal is to minimize the customers' travel times), then the CRRP problem can be {\em decoupled} and the customers and rebalancing flows can be solved separately without loss of optimality. This insight will be instrumental in Section \ref{sec:realtime} to the design of real-time algorithms for routing and rebalancing.

We conclude this section by noticing that the CRRP, from a computational standpoint, can be reduced to an instance of the Minimum-Cost Multi-Commodity Flow problem (Min-MCF), a classic problem in network flow theory \cite{RKA-TLM-JBO:93}. The problem can be efficiently solved either  via  linear programming (the size of the linear program is $|\E| (M+1)$), or via specialized combinatorial algorithms \cite{AVG-ET-RET:89,TL-FM-SP-CS-ET-ST:95, AVG-JDO-SP-CS:98}. However, the solution to the CRRP provides \emph{static fractional} flows, which are not directly implementable for the operation of actual AMoD systems. Practical algorithms (inspired by the theoretical CRRP model) are presented in the next section.

\section{Real-time Congestion-Aware \\Routing and Rebalancing}
\label{sec:realtime}

A natural approach to routing and rebalancing would be to periodically resolve the CRRP  within a receding-horizon, batch-processing scheme (a common scheme for the control of transportation networks  \cite{KTS-NHD-DHL:10, MP-SLS-EF-DR:12, RZ-FR-MP:16EV}). This approach, however, is not directly implementable as the solution to the CRRP provides {\em fractional} flows (as opposed to routes for the {\em individual} vehicles). This shortcoming can be addressed by considering an integral version of the CRRP (dubbed integral CRRP), whereby the flows are {\em integer}-valued and can  be thus easily translated into routes for the individual vehicles, e.g. through a flow decomposition algorithm \cite{LRF-DRF:62}. The integral CRRP, however, is an instance of the 
integral Minimum-Cost Multi-Commodity Flow problem, which is known to be NP-hard \cite{RMK:74, SE-AI-AS:76}. 
Na\"{\i}ve rounding techniques are inapplicable: rounding a solution for the (non-integral) CRRP does not yield, in general, feasible integral flows, and hence feasible routes. For example, continuity of vehicles and customers can not be guaranteed, and vehicles may appear and disappear along a route. In general, to the best of our knowledge, there are no polynomial-time approximation schemes for the integral Minimum-Cost Multi-Commodity Flow problem. 

On the positive side, the integral CRRP admits a decoupling result akin to Theorem \ref{thm:symmetric}: given a set of feasible, {\em integral} customer flows, one can always find a set of feasible, {\em integral} rebalancing flows. (In fact, 
the proof of Theorem \ref{thm:symmetric} does not exploit anywhere the property that the flows are fractional, and thus the proof extends virtually unchanged to the case where the flows are integer-valued). Our approach is to leverage this insight (and more in general the theoretical results from Section \ref{sec:analysis}) to design a heuristic, yet efficient approximation to the integral CRRP that (i) scales to large-scale systems, and (ii) is general, in the sense that can be broadly applied to time-varying, asymmetric networks. 

Specifically, we consider as objective the minimization of the customers' travel times, which, from Section \ref{sec:analysis} and the aforementioned discussion about the generalization of Theorem \ref{thm:symmetric} to integral flows, {\em suggests} that customer routing can be decoupled from vehicle rebalancing (strictly speaking, this statement is only valid for static and symmetric networks -- its generalization beyond these assumptions will be addressed numerically in Section \ref{sec:num}). Accordingly, to emulate the real-world operation of an AMoD system, we divide a given city into geographic regions (also referred to as ``stations'' in some formulations) \cite{MP-SLS-EF-DR:12, RZ-MP:15}, and each arriving customer is assigned the closest vehicle \emph{within that region} (vehicle imbalance across regions is handled separately by the vehicle rebalancing algorithm, discussed below). We apply a greedy, yet computationally-efficient and congestion-aware approach for customer routing where customers are routed to their destinations using the shortest-time path as computed by an $A^*$ algorithm \cite{PEA-NKN-BR:68}. The travel time along each edge is computed using a heuristic delay function that is related to the current volume of traffic on each edge. In this work, for each edge $(u,v)\in \E$ we use the  simple   Bureau of Public Roads (BPR) delay model \cite{BPR:64}

\vspace{-0.5em}
{\small
\begin{equation*}
t_d (u,v) := t(u,v)\left(1 + \alpha \left(\frac{f(u,v)}{c(u,v)}\right)^{\beta}\right),
\end{equation*}%
}%
where $f(u,v) := \sum_{m=1}^{M} \, f_m(u,v) + f_R(u,v)$ is the total flow on edge $(u,v)$, and $\alpha$ and $\beta$ are usually set to $0.15$ and $4$ respectively. Note that customer routing is {\em event-based}, i.e, a routing choice is made as soon as a customer arrives.

Separately from customer routing, vehicle rebalancing from one region to another is performed every $t_{\text{hor}}>0$ time units as a batch process (unlike customer routing, which is an event-based process).
Denote by $v_i(t)$ the number of vehicles in region $i$ at time $t$, and by $v_{ji}(t)$ the number of vehicles traveling from region $j$ to $i$ that will arrive in the next $t_{\text{vicinity}}$ time units. Let $v^{\text{own}}_i(t) := v_i(t) + \sum_j v_{ji}(t)$ be the number of vehicles currently ``owned" by region $i$ (i.e., in the vicinity of such region). Denote by $v^e_i(t)$ the number of excess vehicles in region $i$, or the number of vehicles left after servicing the customers waiting within region $i$. From its definition, $v^e_i(t)$ is given by $v^e_i(t) = v^{\text{own}}_i(t) - c_i(t)$, where $c_i(t)$ is the number of customers within region $i$. Finally, denote by $v^d_i(t)$ the desired number of vehicles within region $i$. For example, for an even distribution of excess vehicles, $v^d_i(t) \propto \sum_i v^e_i(t) / N$, where $N$ is the number regions. Note that the $v^d_i(t)$'s are rounded so they take on integer values. The set of origin regions (i.e., regions that should send out vehicles), $S_R$, and destination regions (i.e., regions that should receive vehicles), $T_R$, for the rebalancing vehicles are then determined by comparing $v^e_i(t)$ and $v^d_i(t)$, specifically, 
\begin{align*}
\text{if } v^e_i(t) &> v^d_i(t), \;\;\;\; \text{region } i \in S_R \\
\text{if } v^e_i(t) &< v^d_i(t), \;\;\;\; \text{region } i \in T_R.
\end{align*}
We assume the residual capacity $c_R(u,v)$ of an edge $(u,v)$, defined as the difference between its overall capacity $c(u,v)$  and the current number of vehicles along that edge, is known and remains approximately constant over the rebalancing time horizon. 
In case the overall rebalancing problem is not feasible (i.e. it is not possible to move all excess vehicles to regions that have a deficit of vehicles while satisfying the congestion constraints), we define slack variables with cost $C$ that allow the optimizer to select a subset of vehicles and rebalancing routes of maximum cardinality such that each link does not become congested. The slack variables are denoted as $ds_i$ for each $i \in S_R$, and $dt_j$ for each $j \in T_R$. 

 Every $t_{\text{hor}}$ time units, the rebalancing vehicle routes are computed by solving the following integer linear program
 
 \vspace{-0.3em}
 {\small
\begin{align}
&\underset{ f_R(\cdot, \cdot), \{ds_i\},  \{dt_j\}}{\text{minimize}} && \! \!\sum_{(u,v) \in \E} t(u,v)\, f_R(u,v)\ + \sum_{i \in S_R} C ds_i + \sum_{i \in T_R} C dt_i \notag \\
&\quad\;\;\text{subject to} && \sum_{u \in \V} f_R(u,v) + 1_{v \in S_R} (v^e_v(t) - v^d_v(t) - ds_v)   \notag \\
& && = \sum_{w \in V}f_R(v,w) + 1_{v \in T_R} (v^d_v(t) - v^e_v(t) - dt_v), \notag \\
& && \qquad\qquad \qquad \qquad\qquad \qquad \qquad \text{for all } v \in \V \notag \\
& && f_R(u,v) \leq c_R(u,v), \;\;\;\; \text{for all }(u,v) \in \E \notag  \\
& && f_R(u,v)  \in \N, \qquad \text{for all } (u,v) \in \E \notag\\
& && ds_i, dt_j \in \N, \qquad  \text{for all }  i \in S_R, j \in T_R \notag
\end{align}
}
The set of (integral) rebalancing flows $\{f_R(u,v)\}_{(u,v)}$ is then decomposed into a set of rebalancing paths  via a flow decomposition algorithm \cite{LRF-DRF:62}. Each rebalancing path connects one origin region with one destination region: thus, rebalancing paths represent the set of routes that excess vehicles should follow to rebalance to  regions with a deficit of vehicles.

The rebalancing optimization problem is an instance of the Minimum Cost Flow problem. If all edge capacities are integral, the linear relaxation of the Minimum Cost Flow problem enjoys a totally unimodular constraint matrix \cite{RKA-TLM-JBO:93}. Hence, the linear relaxation will necessarily have an integer optimal solution, which will be a fortiori an optimal solution to the original Minimum Cost Flow problem. It follows that an integer-valued solution to the rebalancing optimization problem can be computed efficiently, namely in polynomial time, e.g., via linear programming. Several efficient combinatorial algorithms \cite{RKA-TLM-JBO:93} are also available, whose computational performance is typically significantly better.

The favorable computational properties of the routing and rebalancing algorithm presented in this section enable application to large-scale systems, as described next.

\section{Numerical Experiments}
\label{sec:num}

In this section, we characterize the effect of rebalancing on congestion in asymmetric network and  explore the performance of the algorithm presented in Section \ref{sec:realtime} on real-world road topologies with real customer demands.

\subsection{Characterization of Congestion due to Rebalancing in Asymmetric Networks}
\label{sec:congrelax}

The theoretical results in Section \ref{sec:analysis} are proven for capacity-symmetric networks, which are in general a reasonable model for typical urban road networks (we refer the reader to 
\ifarxiv
the Supplementary Material
\else
\cite{RZ-FR-MP:16aEV}
\fi
 for a detailed analysis of capacity symmetry for major U.S. cities).
Nevertheless, 
it is of interest to characterize the applicability of our theoretical results (chiefly, the existential result in Theorem \ref{thm:symmetric}) to road networks that  significantly violate the capacity-symmetry property. In other words, we study to what degree rebalancing might lead to an increase in congestion if the network is asymmetric. 

To this purpose, we compute solutions to the CRRP for road networks with varying degrees of capacity asymmetry and we compare corresponding travel times
to those obtained by computing optimal routes in the absence of rebalancing (as it would be the case, e.g., if the vehicles were privately owned). We focus on the road network portrayed in Figure \ref{fig:Manhattanmap}, which captures all major streets and avenues in Manhattan. 
Transportation requests are based on actual taxi rides in New York City on March 1, 2012 from 6 to 8 p.m. ({courtesy of the New York Taxi and Limousine Commission}). We randomly selected about one third of the trips that occurred in that time frame (roughly 17,000 trips) and we adjusted the capacities of the roads such that the flows induced by these trips would approach the threshold of congestion. The roads considered all have similar speed limits and comparable number of lanes and thus we assign to each edge in the network the same capacity, specifically, one vehicle every 23.6 seconds. This capacity is consistent with the observations that (i) the customer flow is only $30\%$ of the real one (so road capacity is reduced accordingly) and (ii) taxis only contribute to a fraction of the overall traffic in Manhattan. Nevertheless, we stress that the capacity was selected specifically to ensure that the flow induced by the trips would approach the threshold of congestion before any asymmetry is induced.
To investigate the effects of network asymmetry, we  introduce an {\em artificial capacity asymmetry} into the baseline Manhattan road network
by progressively reducing the capacity of all northbound avenues. 

In order to gain a \emph{quantitative} understanding of the effect of rebalancing on congestion and travel times, we introduce slack variables $\delta_C(u,v)$, associated with a cost $c_c(u,v)$, to each congestion constraint \eqref{eqn:pdcong}. The cost $c_c(u,v)$ is selected so that the optimization algorithm will select a congestion-free solution whenever one is available. 
Once a solution is found, the actual travel time on each (possibly congested) link is computed with the heuristic BPR delay model \cite{BPR:64} presented in Section \ref{sec:realtime}. This approach maintains feasibility even in the congested traffic regime, and hence allows us to assess the impact of rebalancing on congestion in asymmetric networks.

Figure \ref{fig:tisims} summarizes the results of our simulations. In the baseline case, no artificial capacity asymmetry is introduced, i.e., the fractional capacity reduction of northbound avenues is equal to 0\%. In this case, the customer routing problem with no rebalancing (essentially, the CRRP problem with the rebalancing flows constrained to be equal to zero) admits a congestion-free solution. On the other hand, the CRRP requires a (very small) relaxation of the congestion constraints. Overall, the difference between the travel times in the two cases is very small and approximately equal to 2.12\%, in line with the fact that New York City's road graph has largely symmetric capacity, as discussed in Section \ref{sec:model} and shown in 
\ifarxiv
the Supplementary Material.
\else
\cite{RZ-FR-MP:16aEV}.
\fi
Interestingly, even with a massive $50\%$ reduction in northbound capacity, travel times when rebalancing vehicles are present are within $4.12\%$ of those obtained assuming no rebalancing is performed. Collectively, these results show that the existential result in Theorem \ref{thm:symmetric}, proven under the assumption of a symmetric network, appears to extend (even though approximately) to asymmetric networks. In particular, it appears that vehicle rebalancing does not lead to an appreciable increase in congestion under very general conditions.

We conclude this section by noticing that for a $40\%$ reduction in capacity, the travel times with vehicle rebalancing dip slightly lower than those without. This effect is due to our use of the BPR link delay model: while in our theoretical model the time required to traverse a link is constant so long as a link is uncongested, the link delay in the BPR model varies by as much as $15\%$ between free-flow and the onset of congestion.

\begin{figure}[t]
\centering
\subfigure{\label{fig:Manhattanmap} \includegraphics[width=0.068\textwidth]{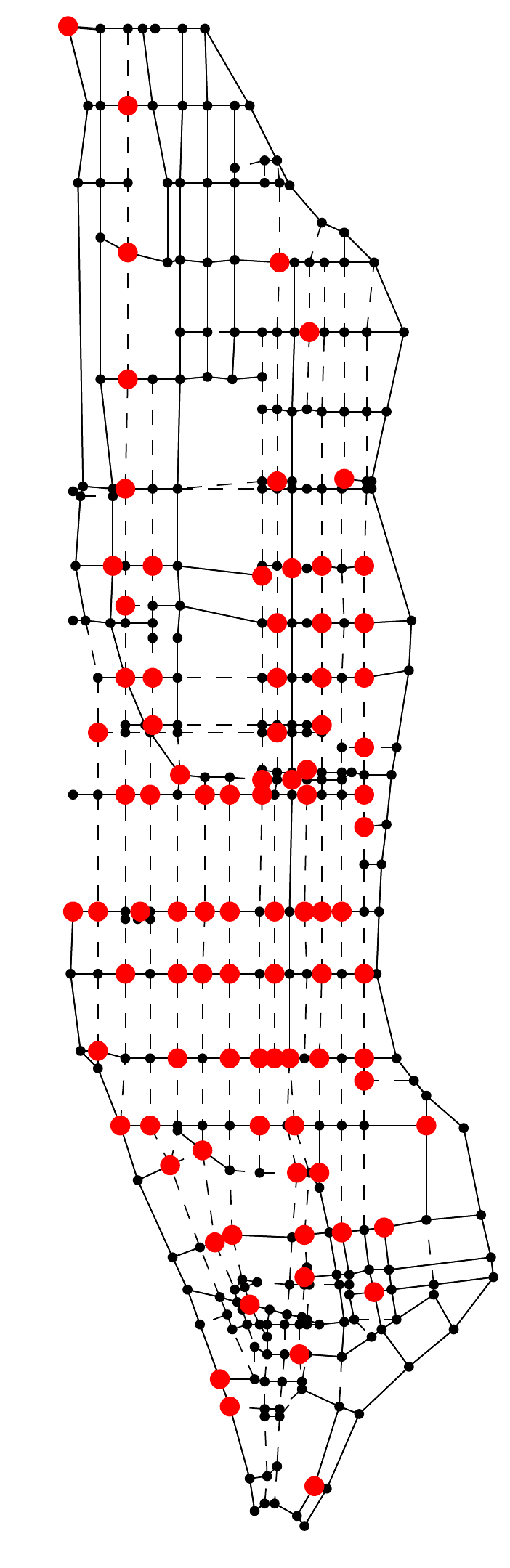}} \hspace{1em}
\subfigure{\label{fig:tisims}
\includegraphics[width=.27\textwidth]{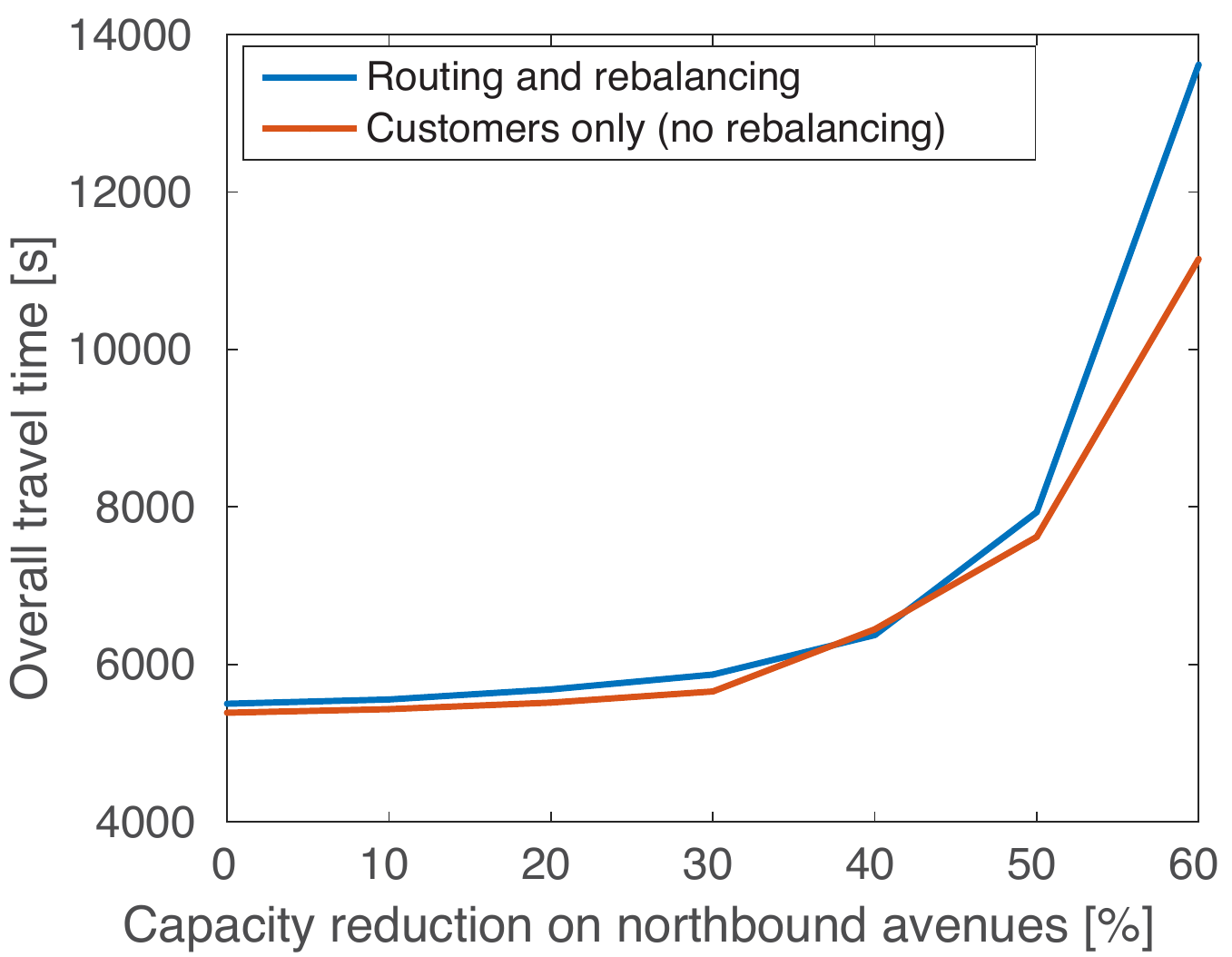}}
\vspace{-1em}
\caption{Left: Manhattan road network. One-way roads are represented as dashed lines. Centers of rebalancing regions are represented in red.
Right: Customer travel times with and without rebalancing for different levels of network asymmetry.}
\end{figure}

\subsection{Congestion-Aware Real-time Rebalancing}
\label{sec:rtsim}

In this section we evaluate the performance of the real-time routing and rebalancing algorithm presented in Section \ref{sec:realtime} against a baseline approach that does not explicitly take congestion into account. We simulate $7,000$ vehicles providing service to actual taxi requests on March 1, 2012, for two hours between 6 and 8 p.m., using the same Manhattan road network as in the previous section (see Figure \ref{fig:Manhattanmap}). Taxi requests are clustered into 88 regions corresponding to a subset of nodes in the road network. Road capacities are reduced to account for exogenous vehicles on the roads to the point that congestion occurs along some routes during the simulation. The free flow speed of the vehicles is set to 25 mph (11 m/s) and approximately 55,000  trip requests (from the taxi data set discussed before) are simulated using a time step of 6 seconds. The simulated speed of the vehicles on each link depend on the number of vehicles in the link, and is calculated using the BPR model. Other delay factors such as traffic signals, turning times, and pedestrian blocking are not simulated. 

Three simulations are performed, namely (i) assuming every customer has access to a private vehicle with no rebalancing, (ii) using the congestion-aware routing and rebalancing algorithm presented in Section \ref{sec:realtime}, and (iii) using a baseline rebalancing algorithm. The baseline approach is derived from the real-time rebalancing algorithm presented in \cite{RZ-MP:15}, which is a point-to-point algorithm that computes rebalancing origins and destinations without considering the underlying road network. In the baseline approach, customer routes are computed in the same way as in Section \ref{sec:realtime}. For rebalancing, the origins and destinations are first solved using the algorithm provided in \cite{RZ-MP:15}, then the routes are computed using the $A^*$ algorithm much like the customer routes. In simulations (ii) and (iii), rebalancing is performed every 2 minutes.

Table \ref{tab:rtsim} presents a summary of the performance results for simulations (ii) and (iii). Note that the service time is the total time a customer spends in the system (waiting $+$ traveling).
\begin{table}[htp]
\begin{center}
\caption{Results of the real-time simulations}\label{tab:rtsim}
\begin{tabular}{r|cc}
Performance metric  & Congestion-aware & Baseline\\
\hline\\ [-8pt]
\# of trips completed & 49,585 & 42,219\\
mean wait time (all trips) & 163.57 s & 406.03 s\\
mean travel time (completed trips) & 265.13 s & 275.19 s \\
mean service time (completed trips) & 286.96 s & 324 s \\
\% with wait time $>$ 5 minutes & 5.4\% & 20\% \\
mean \# of rebalancing vehicles & 204  & 1489
\end{tabular}
\end{center}
\vspace{-.2cm}
\end{table}
Only data from simulations (ii) and (iii) are presented in Table \ref{tab:rtsim} because the only applicable performance metric in simulation (i) is the mean travel time which was 264.69 s. Comparing our algorithm with (i), we notice that the additional rebalancing vehicles have no significant impact on the travel time. Comparing our algorithm with (iii), we notice that the congestion-aware algorithm outperforms the baseline algorithm in every metric: low congestion allows the vehicles to service customers faster, resulting in a reduction in wait times as well as travel times. 
The baseline algorithm will send rebalancing vehicles to stations with a deficit of vehicles regardless of the level of congestion in the road network. This results in many more empty vehicles dispatched to rebalance the system (see Table \ref{tab:rtsim}), which causes heavy congestion in the network\footnote{See the Media Extension, available at \url{https://youtu.be/7OivaJi6CHU}}. Our congestion-aware algorithm drastically reduces this effect, resulting in very few congested road links.

\section{Conclusions and Future Work} \label{sec:conc}
In this paper we presented a network flow model of an autonomous mobility-on-demand system on a capacitated road network. We formulated the routing and rebalancing problem and showed that on symmetric road networks, it is always possible to route rebalancing vehicles in a coordinated way that does not increase traffic congestion. 
Using a model road network of Manhattan, we showed that rebalancing did not increase congestion even for moderate degrees of network asymmetry. We leveraged the theoretical insights to develop a computationally efficient real-time congestion-aware routing and rebalancing algorithm and demonstrated its performance over state-of-the-art point-to-point rebalancing algorithms through simulation. This highlighted the importance of congestion awareness in the design and implementation of control strategies for a fleet of self-driving vehicles.

This work opens the field to many future avenues of research. First, note that the solution to the integral CRRP can directly be used as a practical routing algorithm. For large scale systems, high-quality approximate solutions for the integral CRRP may be obtained using randomized algorithms \cite{PR-CDT:87, AS:99}. Second, from a modeling perspective, we would like to study the inclusion of stochastic information (e.g., demand prediction, travel time uncertainty) for the routing and rebalancing problem, as well as a richer set of performance metrics and constraints (e.g., time windows to pick up customers). Third, it is worthwhile to study how our results give intuition into business models for autonomous urban mobility (e.g. fleet sizes). Fourth, it is of interest to explore other approaches that may reduce congestion, including ride-sharing, demand staggering, and integration with public transit to create an intermodal transportation network. Fifth, we would like to explore decentralized architectures for cooperative routing and rebalancing. Finally, we would like to demonstrate the real-world performance of the algorithms using high fidelity microscopic traffic simulators and by implementing them on real fleets of self-driving vehicles.

\ifarxiv
\section*{Acknowledgement}
The authors would like to thank Zachary Sunberg for his analysis on the road network symmetry of U.S. cities.
\else
\newpage
\fi
\bibliographystyle{IEEEtran-no-url}
{\small
\bibliography{../../../bib/main}

\begin{thebibliography}{10}
\providecommand{\url}[1]{#1}
\csname url@samestyle\endcsname
\providecommand{\newblock}{\relax}
\providecommand{\bibinfo}[2]{#2}
\providecommand{\BIBentrySTDinterwordspacing}{\spaceskip=0pt\relax}
\providecommand{\BIBentryALTinterwordstretchfactor}{4}
\providecommand{\BIBentryALTinterwordspacing}{\spaceskip=\fontdimen2\font plus
\BIBentryALTinterwordstretchfactor\fontdimen3\font minus
  \fontdimen4\font\relax}
\providecommand{\BIBforeignlanguage}[2]{{%
\expandafter\ifx\csname l@#1\endcsname\relax
\typeout{** WARNING: IEEEtran.bst: No hyphenation pattern has been}%
\typeout{** loaded for the language `#1'. Using the pattern for}%
\typeout{** the default language instead.}%
\else
\language=\csname l@#1\endcsname
\fi
#2}}
\providecommand{\BIBdecl}{\relax}
\BIBdecl

\bibitem{WJM-CEBB-LDB:10}
W.~J. {M}itchell, C.~E. {B}orroni {B}ird, and L.~D. {B}urns,
  \emph{\href{http://mitpress.mit.edu/books/reinventing-automobile}{Reinventing
  the Automobile: Personal Urban Mobility for the 21st Century}}.\hskip 1em
  plus 0.5em minus 0.4em\relax Cambridge, MA: The MIT Press, 2010.

\bibitem{Google:14}
{G}oogle,
  ``\href{http://googleblog.blogspot.com/2014/05/just-press-go-designing-self-driving.html}{Just
  Press Go: Designing a Self-Driving Vehicle.}'' Tech. Rep., 2014.

\bibitem{KS-KT-RZ-EF-DM-MP:14}
K.~{S}pieser, K.~{T}releaven, R.~{Z}hang, E.~{F}razzoli, D.~{M}orton, and
  M.~{P}avone, ``{T}oward a {S}ystematic {A}pproach to the {D}esign and
  {E}valuation of {A}utomated {M}obility-{O}n-{D}emand {S}ystems: {A} {C}ase
  {S}tudy in {S}ingapore,'' in \emph{Lecture Notes in Mobility}.\hskip 1em plus
  0.5em minus 0.4em\relax Springer, Jun. 2014, pp. 229--245.

\bibitem{MP-SLS-EF-DR:12}
M.~{P}avone, S.~L. {S}mith, E.~{F}razzoli, and D.~{R}us, ``{R}obotic {L}oad
  {B}alancing for {M}obility-{O}n-{D}emand {S}ystems,'' \emph{International
  Journal of Robotics Research}, vol.~31, no.~7, pp. 839--854, Jun. 2012.

\bibitem{JP-FS-VM-AJ-JCD-TDP:10}
J.~{P}{\'e}rez, F.~{S}eco, V.~{M}ilan{\'e}s, A.~{J}im{\'e}nez, J.~C.
  {D}{\'\i}az, and T.~{D}e {P}edro, ``{A}n {RFID}-based {I}ntelligent {V}ehicle
  {S}peed {C}ontroller {U}sing {A}ctive {T}raffic {S}ignals,''
  \emph{{S}ensors}, vol.~10, no.~6, pp. 5872--5887, 2010.

\bibitem{BT:15}
B.~{T}empleton, ``{T}raffic {C}ongestion \& {C}apacity,'' 2015, available at
  \url{http://www.templetons.com/brad/robocars/congestion.html}.

\bibitem{MB:16}
M.~{B}arnard, ``{A}utonomous {C}ars {L}ikely to {I}ncrease {C}ongestion,'' Jan.
  2016, available at
  \url{http://cleantechnica.com/2016/01/17/autonomous-cars-likely-increase-congestion}.

\bibitem{RZ-MP:15}
R.~{Z}hang and M.~{P}avone, ``{A} {Q}ueueing {N}etwork {A}pproach to the
  {A}nalysis and {C}ontrol of {M}obility-{O}n-{D}emand {S}ystems,'' in
  \emph{{A}merican {C}ontrol {C}onference}, Chicago, IL, Jul. 2015, pp.
  4702--4709.

\bibitem{GB-JFC-GL:10}
\BIBentryALTinterwordspacing
G.~{B}erbeglia, J.-F. {C}ordeau, and G.~{L}aporte, ``{D}ynamic pickup and
  delivery problems,'' \emph{{E}uropean {J}ournal of {O}perational {R}esearch},
  vol. 202, no.~1, pp. 8--15, 2010.
\BIBentrySTDinterwordspacing

\bibitem{KT-MP-EF:11}
\BIBentryALTinterwordspacing
K.~{T}releaven, M.~{P}avone, and E.~{F}razzoli, ``{A}n {A}symptotically
  {O}ptimal {A}lgorithm for {P}ickup and {D}elivery {P}roblems,'' in
  \emph{Proc. {IEEE} Conf. on Decision and Control}, Orlando, FL, Dec. 2011,
  pp. 584--590.
\BIBentrySTDinterwordspacing

\bibitem{KT-MP-EF:13}
\BIBentryALTinterwordspacing
------, ``{A}symptotically {O}ptimal {A}lgorithms for {O}ne-to-{O}ne {P}ickup
  and {D}elivery {P}roblems {W}ith {A}pplications to {T}ransportation
  {S}ystems,'' \emph{IEEE Transactions on Automatic Control}, vol.~58, no.~9,
  pp. 2261--2276, Sep. 2013.
\BIBentrySTDinterwordspacing

\bibitem{KT-MP-EF:12a}
\BIBentryALTinterwordspacing
------, ``{M}odels and {E}fficient {A}lgorithms for {P}ickup and {D}elivery
  {P}roblems on {R}oadmaps,'' in \emph{Proc. {IEEE} Conf. on Decision and
  Control}, Maui, HI, Dec. 2012, pp. 5691--5698.
\BIBentrySTDinterwordspacing

\bibitem{JGW:52}
J.~G. {W}ardrop, ``{S}ome {T}heoretical {A}spects of {R}oad {T}raffic
  {R}esearch,'' in \emph{{ICE} {P}roceedings: {E}ngineering {D}ivisions},
  vol.~1, no.~3.\hskip 1em plus 0.5em minus 0.4em\relax Thomas Telford, 1952,
  pp. 325--362.

\bibitem{MJL-GBW:55}
M.~J. {L}ighthill and G.~B. {W}hitham, ``{O}n kinematic waves. {I}. {F}lood
  movement in long rivers,'' in \emph{{P}roceedings of the {R}oyal {S}ociety of
  {L}ondon {A}: {M}athematical, {P}hysical and {E}ngineering {S}ciences}, vol.
  229, no. 1178.\hskip 1em plus 0.5em minus 0.4em\relax The Royal Society,
  1955, pp. 281--316.

\bibitem{CFD:94}
C.~F. {D}aganzo, ``{T}he cell transmission model: {A} dynamic representation of
  highway traffic consistent with the hydrodynamic theory,''
  \emph{{T}ransportation {R}esearch {P}art {B}: {M}ethodological}, vol.~28,
  no.~4, pp. 269--287, 1994.

\bibitem{BSK:09}
B.~S. {K}erner, \emph{{I}ntroduction to modern traffic flow theory and control:
  the long road to three-phase traffic theory}.\hskip 1em plus 0.5em minus
  0.4em\relax Springer Science \& Business Media, 2009.

\bibitem{MT-AH-DH:00}
M.~{T}reiber, A.~{H}ennecke, and D.~{H}elbing, ``{M}icroscopic simulation of
  congested traffic,'' in \emph{{T}raffic and {G}ranular {F}low 99}.\hskip 1em
  plus 0.5em minus 0.4em\relax Springer, 2000, pp. 365--376.

\bibitem{QY-HNK:96}
Q.~{Y}ang and H.~N. {K}outsopoulos, ``{A} microscopic traffic simulator for
  evaluation of dynamic traffic management systems,'' \emph{{T}ransportation
  {R}esearch {P}art {C}: {E}merging {T}echnologies}, vol.~4, no.~3, pp.
  113--129, 1996.

\bibitem{MB-MR-KM-DC-NL-KN-KA:09}
M.~{B}almer, M.~{R}ieser, K.~{M}eister, D.~{C}harypar, N.~{L}efebvre,
  K.~{N}agel, and K.~{A}xhausen, ``{MATS}im-{T}: {A}rchitecture and simulation
  times,'' \emph{{M}ulti-agent systems for traffic and transportation
  engineering}, pp. 57--78, 2009.

\bibitem{CO-MB:09}
C.~{O}sorio and M.~{B}ierlaire, ``{A}n analytic finite capacity queueing
  network model capturing the propagation of congestion and blocking,''
  \emph{{E}uropean {J}ournal of {O}perational {R}esearch}, vol. 196, no.~3, pp.
  996--1007, 2009.

\bibitem{SP-HSM:95}
S.~{P}eeta and H.~S. {M}ahmassani, ``{S}ystem optimal and user equilibrium
  time-dependent traffic assignment in congested networks,'' \emph{{A}nnals of
  {O}perations {R}esearch}, vol.~60, no.~1, pp. 81--113, 1995.

\bibitem{BNJ:91}
B.~N. {J}anson, ``{D}ynamic traffic assignment for urban road networks,''
  \emph{{T}ransportation {R}esearch {P}art {B}: {M}ethodological}, vol.~25,
  no.~2, pp. 143--161, 1991.

\bibitem{TL-PK-NW-HLV-LLHA-SSPH:15}
T.~{L}e, P.~{K}ov{\'a}cs, N.~{W}alton, H.~L. {V}u, L.~L.~H. {A}ndrew, and
  S.~S.~P. {H}oogendoorn, ``{D}ecentralized signal control for urban road
  networks,'' \emph{{T}ransportation {R}esearch {P}art {C}: {E}merging
  {T}echnologies}, 2015.

\bibitem{NX-EF-YL-YL-YW-DW:15}
N.~{X}iao, E.~{F}razzoli, Y.~{L}uo, Y.~{L}i, Y.~{W}ang, and D.~{W}ang,
  ``{T}hroughput optimality of extended back-pressure traffic signal control
  algorithm,'' in \emph{{C}ontrol and {A}utomation ({MED}), 2015 23th
  {M}editerranean {C}onference on}.\hskip 1em plus 0.5em minus 0.4em\relax
  IEEE, 2015, pp. 1059--1064.

\bibitem{MP-HHS-JMB:91}
M.~{P}apageorgiou, H.~{H}adj {S}alem, and J.-M. {B}losseville, ``{ALINEA}: {A}
  local feedback control law for on-ramp metering,'' \emph{{T}ransportation
  {R}esearch {R}ecord}, no. 1320, 1991.

\bibitem{DW-JPVDB-MCL-DM:11}
D.~{W}ilkie, J.~P. van~den {B}erg, M.~C. {L}in, and D.~{M}anocha,
  ``{S}elf-{A}ware {T}raffic {R}oute {P}lanning,'' in \emph{{AAAI}}, 2011.

\bibitem{DW-CB-MCL:14}
D.~{W}ilkie, C.~{B}aykal, and M.~C. {L}in, ``{P}articipatory {R}oute
  {P}lanning,'' in \emph{{P}roceedings of the 22nd {ACM} {SIGSPATIAL}
  {I}nternational {C}onference on {A}dvances in {G}eographic {I}nformation
  {S}ystems}.\hskip 1em plus 0.5em minus 0.4em\relax ACM, 2014, pp. 213--222.

\bibitem{BSK:09a}
\BIBentryALTinterwordspacing
B.~S. {K}erner, ``\BIBforeignlanguage{English}{{T}raffic {C}ongestion,
  {M}odeling {A}pproaches to},'' in
  \emph{\BIBforeignlanguage{English}{{E}ncyclopedia of {C}omplexity and
  {S}ystems {S}cience}}, R.~A. Meyers, Ed.\hskip 1em plus 0.5em minus
  0.4em\relax Springer New York, 2009, pp. 9302--9355.
\BIBentrySTDinterwordspacing

\bibitem{RKA-TLM-JBO:93}
{R}avindra {K}.~{A}huja, {T}homas {L}.~{M}agnanti, and {J}ames {B}.~{O}rlin,
  \emph{{N}etwork {F}lows: {T}heory, {A}lgorithms and {A}pplications}.\hskip
  1em plus 0.5em minus 0.4em\relax Upper Saddle River, New Jersey 07458:
  Prentice Hall, 1993.

\bibitem{HN:71}
H.~{N}euburger, ``{T}he economics of heavily congested roads,''
  \emph{{T}ransportation {R}esearch}, vol.~5, no.~4, pp. 283 -- 293, 1971.

\bibitem{MH-PW:08}
M.~{H}aklay and P.~{W}eber, ``{O}pen{S}treet{M}ap: {U}ser-{G}enerated {S}treet
  {M}aps,'' \emph{{P}ervasive {C}omputing, {IEEE}}, vol.~7, no.~4, pp. 12--18,
  Oct 2008.

\bibitem{AVG-ET-RET:89}
{A}ndrew {V}.~{G}oldberg, {E}va {T}ardos, and {R}obert {E}.~{T}arjan,
  ``{N}etwork {F}low {A}lgorithms,'' Cornell University Operations Research and
  Industrial Engineering, Tech. Rep., 1989.

\bibitem{TL-FM-SP-CS-ET-ST:95}
T.~{L}eighton, F.~{M}akedon, S.~{P}lotkin, C.~{S}tein, {\'E}.~{T}ardos, and
  S.~{T}ragoudas, ``{F}ast approximation algorithms for multicommodity flow
  problems,'' \emph{{J}ournal of {C}omputer and {S}ystem {S}ciences}, vol.~50,
  no.~2, pp. 228--243, 1995.

\bibitem{AVG-JDO-SP-CS:98}
A.~V. {G}oldberg, J.~D. {O}ldham, S.~{P}lotkin, and C.~{S}tein, ``{A}n
  {I}mplementation of a {C}ombinatorial {A}pproximation {A}lgorithm for
  {M}inimum-{C}ost {M}ulticommodity {F}low,'' in \emph{{I}nteger {P}rogramming
  and {C}ombinatorial {O}ptimization}, ser. Lecture Notes in Computer Science,
  R.~Bixby, E.~Boyd, and R.~R�os-Mercado, Eds.\hskip 1em plus 0.5em minus
  0.4em\relax Springer Berlin Heidelberg, 1998, vol. 1412, pp. 338--352.

\bibitem{KTS-NHD-DHL:10}
K.~T. {S}eow, N.~H. {D}ang, and D.~H. {L}ee, ``{A} collaborative multiagent
  taxi-dispatch system,'' \emph{IEEE Transactions on Automation Sciences and
  Engineering}, vol.~7, no.~3, pp. 607--616, 2010.

\bibitem{RZ-FR-MP:16EV}
R.~{Z}hang, F.~{R}ossi, and M.~{P}avone, ``{M}odel {P}redictive {C}ontrol of
  {A}utonomous {M}obility-on-{D}emand {S}ystems ({E}xtended version),'' Sep.
  2015, available at \url{http://arxiv.org/abs/1509.03985}.

\bibitem{LRF-DRF:62}
L.~R. {F}ord and D.~R. {F}ulkerson, \emph{{F}lows in {N}etworks}.\hskip 1em
  plus 0.5em minus 0.4em\relax Princeton University Press, 1962.

\bibitem{RMK:74}
R.~M. {K}arp, ``{O}n the computational complexity of combinatorial problems,''
  in \emph{{N}etworks, {N}etworks ({USA}), ({P}roceedings of the {S}ymposium on
  {L}arge-{S}cale {N}etworks, {E}vanston, {IL}, {USA}, 18-19 {A}pril 1974.)},
  vol.~5, no.~1, Jan. 1975, pp. 45--68.

\bibitem{SE-AI-AS:76}
\BIBentryALTinterwordspacing
S.~{E}ven, A.~{I}tai, and A.~{S}hamir, ``{O}n the {C}omplexity of {T}imetable
  and {M}ulticommodity {F}low {P}roblems,'' \emph{{SIAM} {J}ournal on
  {C}omputing}, vol.~5, no.~4, pp. 691--703, 1976.
\BIBentrySTDinterwordspacing

\bibitem{PEA-NKN-BR:68}
P.~{H}art, N.~{N}ilsson, and B.~{R}aphael, ``{A} {F}ormal {B}asis for the
  {H}euristic {D}etermination of {M}inimum {C}ost {P}aths,'' \emph{{S}ystems
  {S}cience and {C}ybernetics, {IEEE} {T}ransactions on}, vol.~4, no.~2, pp.
  100--107, July 1968.

\bibitem{BPR:64}
{B}ureau of~{P}ublic {R}oads, ``{T}raffic {A}ssignment {M}anual,'' U.S.
  Department of Commerce, Urban Planning Division, Washington, D.C (1964),
  Tech. Rep., 1964.

\bibitem{PR-CDT:87}
P.~{R}aghavan and C.~D. {T}ompson, ``{R}andomized {R}ounding: {A} {T}echnique
  for {P}rovably {G}ood {A}lgorithms and {A}lgorithmic {P}roofs,''
  \emph{{C}ombinatorica}, vol.~7, no.~4, pp. 365--374, 1987.

\bibitem{AS:99}
A.~{S}rinivasan, ``{A} survey of the role of multicommodity flow and
  randomization in network design and routing,'' \emph{{A}merican
  {M}athematical {S}ociety, {S}eries in {D}iscrete {M}athematics and
  {T}heoretical {C}omputer {S}cience}, vol.~43, pp. 271--302, 1999.

\end{thebibliography}
}
\ifarxiv
\newpage
\ifnamechange
\section*{Appendix: Proofs of Technical Results}
\else
\section*{Supplementary Material: \\Proofs of Technical Results}
\fi
\begin{proof}[Proof of Lemma \ref{lemma:netflow}]
We compute the sum over all customer flows $m\in\mathcal{M}$ and over all nodes $v\in\mathcal{V}$ of the node balance equation for flow $m$ at node $v$ (Equation \eqref{eqn:pdmsource} if node $v$ is the source of $m$, Equation  \eqref{eqn:pdmsink} if node $v$ is the sink of $m$, or Equation \eqref{eqn:pdmbal} otherwise).
We obtain
 \begin{align*}
\sum_{v\in \S} \sum_{m\in \mathcal{M}}&\left(  \sum_{u\in \V} f_m(u,v) + 1_{v=s_m} \lambda_m \right)= \\
&\sum_{v\in \S} \sum_{m\in \mathcal{M}} \left( \sum_{w\in \V} f_m(v,w) + 1_{v=t_m} \lambda_m\right).
 \end{align*}
For any edge $(u,v)$ such that $u,v \in \S$, the customer flow $f_m(u,v)$ appears on both sides of the equation. 
Thus the equation above simplifies to

\begin{align*}
\sum_{m \in \mathcal{M}}\sum_{v\in \S} &\left(\sum_{u\in \bar \S} f_m(u,v) + 1_{v=s_m} \lambda_m\right)=\\
&\sum_{m \in \mathcal{M}}\sum_{v\in \S} \left(\sum_{w\in \bar \S} f_m(v,w) + 1_{v=t_m} \lambda_m\right),
\end{align*}
which leads to the claim of the lemma
\[
F_{\text{in}}(\S,\bar \S) + \sum_{m\in \M} 1_{s_m \in \S} \lambda_m = F_{\text{out}}(\S,\bar \S) + \sum_{m\in \M}1_{t_m \in \S} \lambda_m.
\]
\end{proof}

 \begin{proof}[Proof of Lemma  \ref{lemma:grossflow}]
Adding Equations  \eqref{eqn:pdmbal}, \eqref{eqn:pdmsource} and \eqref{eqn:pdmsink} over all nodes in $\S$ and over all flows whose origin is in $\S$ and whose destination is in $\bar \S$, one obtains 
 \begin{align*}
 \sum_{m: s_m\in \S, t_m \in \bar \S} \, \, \sum_{v\in \S}& \left(\sum_{u\in \V} f_m(u,v) + 1_{v=s_m}\lambda_m\right) =\\
& \sum_{m: s_m\in \S, t_m \in \bar \S}\sum_{v\in \S} \left(\sum_{w\in \V} f_m(v,w) \right).
 \end{align*}
Flows $f_m(u,v)$ such that both $u$ and $v$ are in $\S$ appear on both sides of the equation. Simplifying, one obtains 
 \begin{align*}
&\sum_{m: s_m\in \S, t_m \in \bar \S}  \lambda_m =  \\
&\sum_{m: s_m \in \S, t_m \in \bar \S} \left( \sum_{v\in \S, w \in \bar \S } f_m(v,w) - \sum_{v\in \S, u \in \bar \S } f_m(u,v)\right)
 \end{align*}
 The first term on the right-hand side represents a lower bound for $F_{\text{out}}(\S,\bar \S)$, since
 \begin{align*}
 F_{\text{out}}(\S,\bar \S) =&  \sum_{m \in \mathcal{M}} \sum_{v\in \S, w \in \bar \S } f_m(v,w)
 \\
 \geq&  \sum_{m: s_m \in \S, t_m \in \bar \S} \sum_{v\in \S, w \in \bar \S } f_m(v,w).
 \end{align*}
 Furthermore, the second term on the right-hand side is upper-bounded by zero. The lemma follows.
 \end{proof}

 \begin{proof}[Proof of Lemma \ref{lemma:necessary}]
The first  condition follows trivially from equation  \eqref{eqn:pdcong}.
As for the second condition, consider a cut $(\S, \bar{\S})$. Analogously  as for the definitions of $F_{\text{in}}(\S, \bar{\S})$ and $F_{\text{out}}(\S, \bar{\S})$, let  $F^{\text{reb}}_{\text{in}}(\S, \bar{\S})$  and $F^{\text{reb}}_{\text{out}}(\S, \bar{\S})$ denote, respectively, the overall rebalancing flow entering (exiting) cut  $(\S, \bar{\S})$. 
Summing equation \eqref{eqn:pdrbal} over all nodes in $S$, one easily obtains  
\[
F^{\text{reb}}_{\text{in}}(\S, \bar{\S}) -F^{\text{reb}}_{\text{out}}(\S, \bar{\S}) = \sum_{m \in \M} 1_{s_m\in \S}\lambda_m - \sum_{m\in \M}1_{t_m \in \S}\lambda_m. 
\]
Combining the above equation with Lemma \ref{lemma:netflow}, one obtains
\[
F^{\text{reb}}_{\text{in}}(\S, \bar{\S}) -F^{\text{reb}}_{\text{out}}(\S, \bar{\S})  = F_{\text{out}}(\S, \bar{\S}) - F_{\text{in}}(\S, \bar{\S}),
\]
in other words, rebalancing flows should make up the difference between the  customer inflows and outflows across cut $(\S,\bar \S)$. Accordingly, the  total inflow of vehicles across $(\S, \bar{\S})$, $F^{\text{tot}}_{\text{in}}(\S, \bar{\S})$, satisfies the inequality 
\begin{align*}
F^{\text{tot}}_{\text{in}}(\S, \bar{\S}) :&= F_{\text{in}}(\S,\bar{\S}) + F^{\text{reb}}_{\text{in}}(\S, \bar{\S}) \\
& =  F_{\text{in}}(\S, \bar{\S}) + F^{\text{reb}}_{\text{out}}(\S, \bar{\S})  + F_{\text{out}}(\S, \bar{\S}) - F_{\text{in}}(\S, \bar{\S}) \\
&\geq F_{\text{out}}(\S, \bar{\S}).
\end{align*}
Since the  customer and rebalancing flows $\{f_m(u,v), \, f_R(u,v)\}_{(u,v),m}$ are feasible, then, by equation \eqref{eqn:pdcong}, $F^{\text{tot}}_{\text{in}}(S, \bar{S}) \leq C_\text{in}(S, \bar{S})$. Collecting the results, one obtains the second condition.
\end{proof}

 \begin{proof}[Proof of Lemma \ref{lemma:prsourcesink}]
By contradiction. Since the flow $\{\hat f_R(u,v)\}_{(u,v)}$ is not a feasible rebalancing flow, there exists at least one defective origin or a defective destination. Assume that there exists at least one defective destination, say a node  $\hat t_j$ where Equation \eqref{eqn:pdrbal} is violated:
\[
\sum_{u\in \V} \hat f_R(u,\hat t_j) +\sum_{m\in \mathcal M} 1_{\hat t_j = t_m}\lambda_m 
>  \sum_{w\in \V} \hat f_R(\hat t_j,w),
\]
Now, assume that there does not exist any defective origin. By summing Equation \eqref{eqn:pdrbal} over all nodes $v\in \V$ and simplifying all flows $\hat f_R(u,v)$ (as they appear on both sides of the resulting equation), one obtains
\begin{align*}
\sum_{v\in \V}\sum_{m\in \M} 1_{v=t_m} \lambda_m >& \sum_{v\in \V}\sum_{m\in \M} 1_{v=s_m} \lambda_m,
\end{align*}
that is $\sum_{m\in \M}\lambda_m >  \sum_{m\in \M} \lambda_m$, which is a contradiction. Noticing that the symmetric case where we assume that there exists at least one defective destination leads to an analogous contradiction, the lemma follows.
\end{proof}

\begin{proof}[Proof of Lemma \ref{lemma:satcut}]
The proof is constructive and constructs the desired partial rebalancing flow by starting with the trivial zero flow $\hat f_R(u,v)=0$ for all $(u,v)\in \E$. Let $\Vdo:=\{\hat s_1, \ldots, \hat s_{|\Vdo|}\}$ and $\Vdd:=\{\hat t_1, \ldots, \hat t_{|\Vdd|}\}$ be the set of defective origins and destinations, respectively, under such flow. Then, the zero flow is iteratively updated according to the following procedure:
\begin{enumerate}
\item Look for a path between a node in $\Vdd$ and a node in $\Vdo$ that is not saturated (note that for rebalancing flows, paths go from  destinations to origins). If no such path exists, quit. Otherwise, go to Step 2. 
\item Add the same amount of flow on all edges along the path until either (i) one of the edges becomes saturated or (ii) constraint \eqref{eqn:pdrbal} is fulfilled either at the defective origin or at the defective destination. Note that the resulting flow remains a partial rebalancing flow. 
\item Update sets $\Vdo$ and $\Vdd$ for the new partial rebalancing flow and go to Step 1.
\end{enumerate}

The algorithm terminates. To show this, we prove the invariant that if a node is no longer defective for the updated partial rebalancing  flow (in other words, Step 2 ends due to condition (ii)), it will not become defective at a later stage. Consider a defective destination node $v$ that becomes non-defective under the updated partial rebalancing flow (the proof for defective origins is analogous). Then, at the subsequent stage it cannot be considered as a destination in Step 1 (as it is no longer in set $\Vdd$). If a path that does not contain $v$ is selected, then $v$ stays non-defective. Otherwise, if a path that contains $v$ is selected, then, after Step 2, both the inbound flow (that is the flow into $v$) and the outbound flow (that is the flow out of $v$) will be increased by the same quantity, and the node will stay non-defective. An induction on the stages then proves the claim. As the number of paths is finite, and sets $\Vdo$ and $\Vdd$ cannot have any nodes added, the algorithm terminates after a finite number of stages.

The output of the algorithm (denoted, with a slight abuse of notation, as $\{\hat f_R(u,v)\}_{(u,v)}$) is a partial rebalancing flow that is not feasible (as, by assumption, there does not exist a set of
feasible rebalancing flows).  Therefore, by Lemma \ref{lemma:prsourcesink}, such partial rebalancing flow has at least one defective origin and at least one defective destination. 
Let us define $\E_{ns}:=\E\setminus\{(u,v):(u,v)\text{ is saturated}\}$ as the collection of non-saturated edges under the flows $\{f_m(u,v)\}_{(u,v),m}$ and $\{\hat f_R(u,v)\}_{(u,v)}$. For any defective destination  and any defective origin, all paths connecting them contain at least one saturated edge (due to the exit condition in Step 1).
Therefore, the graph $G_{ns}(\V, \E_{ns})$ has two properties: (i) it is disconnected (that is, it is not possible to find a direct path between every pair of nodes in $\V$ by using edges in $\E_{ns}$), and (ii) a defective origin and a defective destination can not be in the same strongly connected component (hence, graph $G_{ns}(\V, \E_{ns})$ can be partitioned into at least two strongly connected components).

We now find the cut $(\S,\bar \S)$ as follows. If a strongly connected component of $G_{ns}$ contains defective destinations, we assign its nodes to set $\S$. If a strongly connected component contains defective origins, we assign its nodes to set $\bar \S$. If a strongly connected component contains neither defective origins nor destinations, we assign its nodes to $\S$ (one could also assign its nodes to $\bar \S$, but such choice is immaterial for our purposes). By construction, $(\S, \bar \S)$ is a cut, and its edges are all saturated. Furthermore, set $\S$ only contains destination nodes, and set $\bar \S$ only contains origin nodes, which concludes the proof.
\end{proof}

\ifnamechange
\section*{Appendix: Capacity Symmetry within Urban Centers in the US}
\else
\section*{Supplementary Material: Capacity Symmetry within Urban Centers in the US}
\fi
\label{sec:unbalanceOSM}

The existential result in Section \ref{sec:analysis}, Theorem \ref{thm:symmetric}, relies on the assumption that the road network is capacity-symmetric, i.e., for every cut $(\S,\bar \S), C_\text{out}(\S,\bar \S)=C_\text{in}(\S,\bar \S)$.
One may wonder whether this assumption is (approximately) met in practice. From an intuitive standpoint, one might argue that transportation networks within urban centers are indeed {\em designed} to be capacity symmetric, so as to avoid accumulation of traffic flow in some directions. We corroborate this intuition by computing the imbalance between the outbound capacity (i.e., $C_\text{out}$) and the inbound capacity (i.e., $C_\text{in}$) for 1000 randomly-selected cuts within several urban centers in the United States. For each edge $(u,v)\in \E$, we approximate its capacity as proportional to the product of the speed limit $v_{\text{max}}(u,v)$ on that edge and the number of lanes $L(u,v)$, that is,
$c(u,v)\propto v_{\text{max}}(u,v) \cdot L(u,v)$.
The road graph $G(\V,\E)$, the speed limits, and the number of lanes are obtained from OpenStreetMap data \cite{MH-PW:08}.

For a cut $(S,\bar S)$, we define its fractional capacity disparity $D(\S,\bar{\S})$ as
\[
D(\S,\bar{\S}):=2\, \frac{\left|C_\text{out}(\S,\bar \S) -C_\text{in}(\S,\bar \S)\right|}{C_\text{out}(\S,\bar \S) +C_\text{in}(\S,\bar \S)}.
\]
Table \ref{tab:unbalanceOSM} shows the average (over 1000 samples) fractional capacity disparity for several US urban centers. As expected, the road networks for such cities appear to posses a very high degree of capacity-symmetry, which validates the symmetry assumption made in Section \ref{sec:analysis}. 

\begin{table}[htp]
\begin{center}\label{tab:unbalanceOSM}
\caption{Average fractional capacity disparity for several major urban centers in the United States.}
\begin{tabular}{rcc}
Urban center & Avg. frac. capacity disparity & Std. dev.\\
\hline\\ [-6pt]
Chicago, IL & 1.2972 $\cdot 10^{-4}$& $1.003\cdot 10^{-4}$\\
New York, NY & 1.6556 $\cdot 10^{-4}$&$1.304 \cdot 10^{-4}$\\
Colorado Springs, CO &3.1772 $\cdot 10^{-4}$&$ 2.308\cdot 10^{-4}$\\
Los Angeles, CA & 0.9233 $\cdot 10^{-4}$&$0.676\cdot 10^{-4}$\\
Mobile, AL & 1.9368 $\cdot 10^{-4}$& $1.452\cdot 10^{-4}$\\
Portland, OR & 1.0769 $\cdot 10^{-4}$& $0.778\cdot 10^{-4}$\\
\end{tabular}
\end{center}
\end{table}

\fi
\end{document}